\documentclass[12pt,twoside]{article} 
\usepackage[T1]{fontenc} 

\usepackage{amsthm, amsmath, amssymb, mathrsfs, upgreek, enumerate, mathtools, comment, natbib, graphicx, mathabx, booktabs, threeparttable, tikz, relsize, exscale, epstopdf, scrextend, multirow, xr-hyper, pdfpages, bm, caption, subcaption}
\usepackage[headheight=110pt,margin=1.25in]{geometry}
\usepackage[algoruled]{algorithm2e}
\usepackage[section]{placeins} 

\usepackage{fancyhdr, setspace} 
\usepackage[pagebackref=false]{hyperref} 
\hypersetup{
    colorlinks=true,
    citecolor=blue,
    filecolor=black,
    linkcolor=black,
    urlcolor=blue,
    bookmarksopen=true,
    pdfstartview=FitH
}

\usepackage[labelsep=period]{caption}
\newtheorem{theorem}{Theorem}
\newtheorem{assump}{Assumption}

\newtheorem{proposition}{Proposition}

\theoremstyle{definition}
\newtheorem{definition}{Definition}

\newtheorem{example}{Example}

\newtheorem{remark}{Remark}

\def\Snospace~{\S{}}

\makeatletter
\def\thm@space@setup{
  \thm@preskip=15pt \thm@postskip=15pt 
}
\makeatother

\def\indep{\perp\!\!\!\perp}

\newcommand{\argmin}{\operatornamewithlimits{argmin}}

\newcommand{\cov}{\text{Cov}}

\newcommand{\E}{{\bf E}}
\newcommand{\R}{\mathbb{R}}
\newcommand{\C}{\mathcal{C}}
\newcommand{\N}{\mathcal{N}}
\newcommand{\prob}{{\bf P}}
\newcommand{\plimarrow}{\stackrel{p}\longrightarrow}
\newcommand{\dlimarrow}{\stackrel{d}\longrightarrow}
\newcommand{\ind}{\bm{1}}
\newcommand{\zero}{\bm{0}}
\providecommand{\abs}[1]{\lvert#1\rvert} 
\providecommand{\norm}[1]{\lVert#1\rVert}

\newcommand*{\medcup}{\mathbin{\scalebox{1.5}{\ensuremath{\cup}}}}

\let\emptyset\varnothing
\providecommand{\abs}[1]{\lvert#1\rvert} 
\providecommand{\norm}[1]{\lVert#1\rVert}

\renewcommand{\qed}{\hfill \mbox{\raggedright \rule{0.08in}{0.08in}}} 
\renewenvironment{proof}[1][\proofname]{{\noindent\sc#1. }}{\qed\vspace{15pt}} 

\title{\bf\sc Network Cluster-Robust Inference}

\author{Michael P.\ Leung\thanks{Department of Economics, UC Santa Cruz. E-mail: leungm@ucsc.edu. I thank the referees for suggestions that improved the exposition of the paper. I also thank seminar audiences at CUHK, Harvard/MIT, Sciences Po, the Stanford RAIN seminar, UCLA, U.\ Geneva, U.\ Penn., and U.\ Wisconsin for helpful comments. Part of this research was conducted at the University of Southern California.}}

\begin{document}
\maketitle
\onehalfspacing
 
\begin{abstract}

  {\sc Abstract.} Since network data commonly consists of observations from a single large network, researchers often partition the network into clusters in order to apply cluster-robust inference methods. Existing such methods require clusters to be asymptotically independent. Under mild conditions, we prove that, for this requirement to hold for network-dependent data, it is necessary and sufficient that clusters have low conductance, the ratio of edge boundary size to volume. This yields a simple measure of cluster quality. We find in simulations that when clusters have low conductance, cluster-robust methods control size better than HAC estimators. However, for important classes of networks lacking low-conductance clusters, the former can exhibit substantial size distortion. To determine the number of low-conductance clusters and construct them, we draw on results in spectral graph theory that connect conductance to the spectrum of the graph Laplacian. Based on these results, we propose to use the spectrum to determine the number of low-conductance clusters and spectral clustering to construct them. 
  
  \bigskip 


 
\end{abstract}

\newpage

\section{Introduction}\label{sintro}

Cluster-robust methods are widely used to account for cross-sectional dependence \citep{cameron2015practitioner,hansen2019asymptotic}. The standard model of cluster dependence partitions the set of observations into many independent clusters, but there are common settings in which observations cannot be so divided. Spatially or temporally dependent data often have the property that the correlation between observations decays with distance but is never precisely zero, and recent econometric work develops cluster-robust methods applicable to data of this type \citep{bester2011inference,canay2017randomization,canay2020wild,ibragimov2010t,ibragimov2016inference}. This paper studies the performance of these methods under network-mediated dependence.

Researchers commonly observe only a single large network, for instance the friendship network of a school or village, and clustered standard errors are frequently used to account for network dependence. However, it is often unclear how best to partition a network into clusters for the purpose of inference. Several papers employ ``community detection'' or ``network clustering'' unsupervised learning algorithms to divide the network into disjoint subnetworks according to some criteria and then cluster standard errors on the subnetworks \citep[e.g.][]{aral2017exercise,aral2019social,eckles2016estimating,zacchia2020knowledge}. However, the objectives of unsupervised learning algorithms do not necessarily align with the econometric objective of inference. Moreover, the algorithms typically depend on tuning parameters that can be chosen to mechanically increase the number of clusters returned, and there are no guidelines on how to choose this number for the purpose of inference. 

An alternative is to use HAC variance estimators, which are widely used to account for spatial or temporal dependence. The same estimators may be used for network dependence, essentially by interchanging spatial or temporal distance with network (path) distance \citep{kojevnikov2021bootstrap,kojevnikov2021limit}. However, simulation evidence has shown that for spatially or temporally dependent data, tests using HAC estimators can exhibit greater size distortion than cluster-robust inference methods in smaller samples \citep{conley2018inference,ibragimov2010t}. This is also born out in our simulations with network data but with significant caveats discussed below.

To our knowledge, there is presently no theoretical justification for applying cluster-robust methods to network-dependent data. \cite{eckles2016estimating} and \cite{zacchia2020knowledge} invoke the work of \cite{bester2011inference} to justify network clustering, but the latter's results are specific to spatial data, which, as we will discuss, have different properties in terms of the existence of sufficiently segregated clusters. Regarding HAC estimation, there are important differences in the choice of the bandwidth for data that is network, rather than spatially, dependent \citep{kojevnikov2021limit,leung2020causal}. A motivating question for this paper is whether cluster-robust methods also encounter novel complications when applied to network-dependent data.

\bigskip
\noindent {\bf Contributions.} We show that complications do exist. Whereas weakly dependent spatial data can always be partitioned into several ``quality'' clusters (we will define what we mean by ``quality''), multiple such clusters are not guaranteed to exist in a given network. We show that cluster-robust methods applied to networks that lack quality clusters can exhibit substantial size distortion. This motivates the methods provided in this paper for diagnosing whether quality clusters exist and determining how to construct them.\footnote{A Python package implementing the methods is available at \url{https://pypi.org/project/networkinference/}. \label{networkinference}}

We derive conditions on the clusters and data-generating process under which cluster-robust methods are valid under network dependence. Existing such methods require asymptotically independent clusters. \cite{bester2011inference} provide primitive sufficient conditions for asymptotic independence under spatial dependence, the key assumption being that cluster boundary sizes grow at a sufficiently slow rate relative to the sample size. For network data, we show that, under certain conditions, it is necessary and sufficient that clusters have low {\em conductance}, which is roughly the proportion of a cluster's links sent outside the cluster (formally defined in \autoref{ssum}).\footnote{Low conductance does {\em not} require that clusters be asymptotically disconnected because the number of between-cluster links can diverge.} 
This yields a simple $[0,1]$-measure of cluster quality and suggests a (computationally infeasible) objective for constructing clusters: choose the set that minimizes conductance. The result connects the literature on cluster-robust inference to spectral graph theory and spectral clustering, which we draw upon to feasibly construct clusters.

Due to the topology of Euclidean space, clusters satisfying the \cite{bester2011inference} boundary rate condition can always be constructed for spatial data under increasing domain asymptotics. However, we show that in the network setting this may not be possible depending on the network formation process. As we discuss, for low-conductance clusters to exist, the network must have a sufficiently small (higher-order) {\em Cheeger constant}, which is a well-known graph invariant that measures network integration. Some classes of networks satisfy this condition, but many apparently do not.\footnote{A trivial example is a fully connected network, but we will discuss more realistic examples.} We therefore require methods to determine the existence of low-conductance clusters and, if they exist, to construct them.

Computing the Cheeger constant or set of clusters minimizing conductance turns out to be computationally infeasible. Fortunately, Cheeger inequalities imply that the lower eigenvalues of the graph Laplacian (defined in \autoref{sginvs}) are informative about the constant's magnitude. A simple argument (\autoref{easy}) shows that a set of $L$ low-conductance clusters exists if and only if the $L$th smallest eigenvalue is small, providing a practical diagnostic for determining both the existence and number of quality clusters. To then compute the clusters, we can apply $k$-means clustering to the eigenvectors, which corresponds to spectral clustering, a widely used unsupervised learning algorithm.

Our simulation results show that there are advantages to using cluster-robust methods for network data, relative to HAC estimators. We find that the randomization test of \cite{canay2017randomization}, a leading method for cluster-robust inference with a small number of clusters, better controls size in smaller samples, provided clusters have low conductance. However, when no such clusters exist, the test, when naively applied to the output of spectral clustering, can exhibit substantial size distortion even in large samples, unlike the HAC estimator. This is because clusters in this case cannot generally satisfy the requirement of asymptotic independence, so we expect all existing cluster-robust methods to exhibit similar size distortion.

Based on these results, we make three recommendations for empirical practice in \autoref{srecs}. These concern how to assess whether a given set of clusters is of sufficient ``quality'' (compute the conductance), how to assess whether quality clusters exist (compute the spectrum of the Laplacian), and how to compute such clusters if they exist (apply spectral clustering or possibly other community detection algorithms). We provide a Python package for implementing these steps, linked in \autoref{networkinference}.

Networks often contain only a small number of low-conductance clusters, as seen in our simulation results, and conventional clustered standard errors perform poorly with few clusters \citep{cameron2015practitioner}. Our results therefore utilize asymptotics sending the sizes of a fixed number of clusters to infinity, the same framework adopted by the literature on cluster-robust methods for few clusters \citep{bester2011inference,canay2017randomization,canay2020wild,ibragimov2010t,ibragimov2016inference}. Our formal results provide interpretable primitive conditions under which the key independence assumption imposed by these papers holds. We do not develop a new inference procedure; rather, we provide diagnostics to assess whether these existing procedures are valid when applied to network data and an asymptotic theory supporting the diagnostics.

\bigskip
\noindent {\bf Related Literature.} There is little work on how to best construct clusters even for spatial data \citep{conley2018inference}. Partitioning space into equally sized rectangles satisfies the boundary condition of \cite{bester2011inference}, but with irregularly spaced data, it may be possible to do better. \cite{ibragimov2016inference} provide a test for whether a finer cluster partition is the correct level of clustering compared to a coarser one. \cite{mueller2021spatial} propose novel standard errors for spatial dependence using the principal components of a benchmark correlation model.

\cite{jochmans2019fixed} demonstrate the utility of the spectrum for assessing the precision of fixed-effect network regressions. Our results highlight another practical use of the spectrum, namely to assess the validity of cluster-robust methods.

There is a large, cross-disciplinary literature on spectral clustering. \cite{von2007tutorial} is a well-known reference in computer science. A growing literature in statistics studies the performance of spectral clustering for estimating stochastic block models \citep[e.g.][]{lei2015consistency,rohe2011spectral}. The goal of this literature is to recover latent types (``communities''), and theoretical results establish convergence of the sample Laplacian to a population analog that identifies types. We are instead interested in spectral clustering from the graph conductance perspective of identifying small-boundary clusters, regardless of their association or lack thereof with a structural parameter such as type. This is significant because sufficiently dense graphs are required to identify types, whereas we study sparse graphs, which are more common in practice. The conductance perspective seems to be emphasized more in the computer science literature \citep[see e.g.][]{trevisan2016lecture}.

In the next section, we present the model. In \autoref{sov}, we summarize our main results and their intuition and make recommendations for empirical practice. We present two empirical illustrations in \autoref{sapp} and the asymptotic theory in \autoref{smain}. We then provide theoretical and simulation evidence in \autoref{sspectra} showing that important classes of graphs lack quality clusters. In \autoref{smc}, we present simulation results on the finite-sample properties of cluster-robust methods and HAC estimators. Finally, \autoref{sconclude} concludes.

\section{Setup}\label{smodel}

We observe a set of units $\mathcal{N}_n = \{1,\dots,n\}$, data $W_i \in \R^{d_w}$ associated with each unit $i\in\mathcal{N}_n$, and an undirected network or graph $\bm{A}$ on $\mathcal{N}_n$. We represent $\bm{A}$ as a binary, symmetric adjacency matrix with $ij$th entry $A_{ij}$, where $A_{ij}=1$ (0) signifies (the absence of) a link between $i$ and $j$. There are no self-links, meaning $A_{ii}=0$ for all $i$. \autoref{rwdg} below discusses possible extensions to weighted, directed networks.

Our analysis treats $\bm{A}$ as fixed (conditioned upon), whereas $\{W_i\}_{i=1}^n$ is random and not necessarily identically distributed. Let $\theta_0 \in \R^{d_\theta}$ be the estimand of interest and $g\colon \R^{d_w} \times \R^{d_\theta} \rightarrow \R^{d_g}$ a moment function such that
\begin{equation*}
  \E[g(W_i,\theta_0)]=\zero \quad\forall\, i \in \mathcal{N}_n.\footnote{More generally, we could relax this to $\abs{\C}^{-1} \sum_{i\in\C} \E[g(W_i,\theta_0)]=\zero$ for each cluster (defined below) $\C \subseteq \mathcal{N}_n$. Existing cluster-robust methods for a small number of clusters require homogeneity of $\theta_0$ across clusters, in which case this moment condition would not hold in design-based settings where $\theta_0$ corresponds to a finite-population average over all units in the network.} 
\end{equation*}

\noindent Let $\hat G(\theta) = n^{-1} \sum_{i=1}^n g(W_i,\theta)$ and $\Psi_n$ be a weighting matrix, and define the generalized method of moments (GMM) estimator
\begin{equation*}
  \hat\theta = \argmin_\theta \hat G(\theta)' \Psi_n \hat G(\theta).
\end{equation*}

\noindent For example, to recover parameters of linear-in-means models, \cite{aral2017exercise}, \cite{eckles2016estimating}, and \cite{zacchia2020knowledge} use instrumental variables estimators, which are well-known special cases of GMM.

Various papers cited in \autoref{sintro} develop cluster-robust methods for GMM when $\{W_i\}_{i=1}^n$ satisfies weak spatial or temporal dependence. We instead employ a notion of weak network dependence due to \cite{kojevnikov2021limit}, formally defined in \autoref{smain}, which is analogous to mixing conditions used in spatial and time series econometrics. The key difference is the metric, which is path distance. For any two units $i,j$, their {\em path distance} $\ell_{\bm{A}}(i,j)$ is the length of the shortest path between them in $\bm{A}$.\footnote{A {\em path} between $i,j$ is a sequence of links $A_{k_1k_2}, A_{k_2k_3}, \dots, A_{k_{m-1}k_m}=1$ such that $k_1=i$, $k_m=j$, and $k_a \neq k_b$ for all $a,b \in \{1, \dots, m\}$. The {\em length} $\ell_{\bm{A}}(i,j)$ of this path is $m-1$. If $i\neq j$ and no path between $i,j$ exists, then we define $\ell_{\bm{A}}(i,j)=\infty$. If $i=j$, we define $\ell_{\bm{A}}(i,j)=0$.} Weak network dependence essentially demands that the correlation between $W_i$ and $W_j$ decays to zero as $\ell_{\bm{A}}(i,j) \rightarrow \infty$. 

\subsection{Clusters}\label{sclusters}

Cluster-robust methods take as input a partition of $\mathcal{N}_n$ into $L$ clusters, which we denote by $\{\C_\ell\}_{\ell=1}^L$. Being a partition, the clusters satisfy $\medcup_{\ell=1}^L \C_\ell = \mathcal{N}_n$ and $\C_\ell \cap \C_m = \emptyset$ for all $\ell\neq m$. Note that each $\C_\ell$ usually depends on $\bm{A}$ since different networks may be partitioned differently.

As mentioned in \autoref{sintro}, the number of quality clusters in a network is often small, so our formal results are specific to inference procedures robust to a small number of clusters, discussed in the next subsection. When referring specifically to such procedures, we use the label ``small-$L$ cluster-robust methods.'' Such methods have been shown to exhibit substantially improved size control relative to conventional cluster-robust procedures \citep{cameron2015practitioner}. We expect that our results remain relevant in settings with many clusters. 

A network may consist of multiple {\em components}, which are connected subnetworks that are disconnected (in the sense of infinite path distance) from the rest of the network. Under weak network dependence, observations in different components are independent, so components may therefore be treated as separate clusters. This implies that, if a network consists of many components, standard many-cluster asymptotics are applicable, and one can simply cluster standard errors on the components. 

However, a well-known stylized fact about real-world networks is that they typically possess a {\em giant component} containing the vast majority of units (formally, order $n$ units), and all other components are small \citep[formally, being $o(n)$ and typically $O(\log n)$ in size; see][]{barabasi2015}. For example, the giant component of the Facebook graph contains 99.91 percent of all units, whereas its second-largest component only has about 2000 units \citep{ugander2011anatomy}. Therefore, the key task for clustering is partitioning the giant component. All other components, being small, can be treated as individual clusters or grouped into a single cluster.

\subsection{Inference Procedures}\label{sinferp}

Suppose we have a set of clusters $\{\C_\ell\}_{\ell=1}^L$. Let $n_\ell = \abs{\C_\ell}$ be the cardinality of $\C_\ell$, $\hat\theta_\ell$ the GMM estimator computed using only observations in $\C_\ell$, and $\hat{G}_\ell(\theta) = n_\ell^{-1} \sum_{i\in\C_\ell} g(W_i,\theta)$ the sample moments constructed using only observations in $\C_\ell$. Small-$L$ cluster-robust methods use estimates $(\hat\theta_\ell)_{\ell=1}^L$ or moments $(\hat G_\ell(\hat\theta_\ell))_{\ell=1}^L$ (possibly for constrained versions of $\hat\theta_\ell$) to construct tests. A popular method is the wild bootstrap procedure due to \cite{cameron2008bootstrap}, whose formal properties under fixed-$L$ asymptotics are studied by \cite{canay2020wild}. Their results, as well as those of \cite{bester2011inference}, require clusters to satisfy certain homogeneity conditions, which are not imposed by \cite{canay2017randomization} and \cite{ibragimov2010t}.

\cite{cai2022implementation} argue that the randomization test of \cite{canay2017randomization} has several attractive properties relative to the other alternatives. We next define their test since it will be the focus of our simulation study in \autoref{smc}. Let $S_n = ( \sqrt{n}(\hat\theta_\ell-\theta) )_{\ell=1}^L$, and define the Wald statistic
\begin{equation*}
  T(S_n) = \left( \frac{1}{\sqrt{L}} \sum_{\ell=1}^L \sqrt{n}(\hat\theta_\ell-\theta)' \right) \left( \frac{1}{L} \sum_{\ell=1}^L n(\hat\theta_\ell-\theta)(\hat\theta_\ell-\theta)' \right)^{-1} \left( \frac{1}{\sqrt{L}} \sum_{\ell=1}^L \sqrt{n}(\hat\theta_\ell-\theta) \right),
\end{equation*}

\noindent where $\theta$ is the null value of $\theta_0$. This effectively treats $(\hat\theta_\ell)_{\ell=1}^L$ as $L$ independent observations, hence why asymptotic independence is required. To construct critical values, we generate a ``randomization distribution'' $\{T(\pi S_n)\colon \pi \in \{-1,1\}^L\}$ where $\pi S_n = ( \pi_\ell \sqrt{n}(\hat\theta_\ell-\theta) )_{\ell=1}^L$ for $\pi = (\pi_\ell)_{\ell=1}^L \in \{-1,1\}^L$. Let $\alpha$ be the level of the test, $k = \lceil 2^L(1-\alpha) \rceil$, and $T^{(k)}(S_n)$ be the $k$th largest value of $\{T(\pi S_n)\colon \pi \in \{-1,1\}^L\}$, a quantile of the randomization distribution. The test rejects if 
\begin{equation}
  T(S_n) > T^{(k)}(S_n). \label{sinfp}
\end{equation}

\noindent A confidence region can be obtained by inverting the test, automated in the package linked in \autoref{networkinference} as well as in R and Stata packages due to \cite{cai2022implementation}.

\section{Summary of Results and Recommendations}\label{sov}

We next define key concepts, summarize our theoretical results, and make practical recommendations for constructing clusters and assessing their quality.

\subsection{Conductance}\label{ssum}

In \autoref{smain}, we provide formal conditions under which a given set of clusters can be used for asymptotically valid cluster-robust inference. Here we provide an overview of the main ideas. We consider a sequence of networks with associated clusters indexed by the network size $n$, taking $n$ to infinity while keeping the number of clusters $L$ fixed. For economy of language, we often simply refer to a network rather than a sequence of networks when discussing asymptotic results.

Under weak network dependence and standard regularity conditions, we show that
\begin{equation}
  \resizebox{\columnwidth}{!}{
  $\frac{1}{\sqrt{n}} \begin{pmatrix} n_1 \hat{G}_1(\theta_0) \\ \vdots \\ n_L \hat{G}_L(\theta_0) \end{pmatrix} \dlimarrow \mathcal{N}(\zero,\bm{\Sigma}^*), \quad \bm{\Sigma}^* = \begin{pmatrix} \rho_1\bm{\Sigma}_{11} & \sqrt{\rho_1\rho_2}\, \bm{\Sigma}_{12} & \dots & \sqrt{\rho_1\rho_L}\, \bm{\Sigma}_{1L} \\ \sqrt{\rho_2\rho_1}\, \bm{\Sigma}_{21} & \rho_2\bm{\Sigma}_{22} & \dots & \sqrt{\rho_2\rho_L}\, \bm{\Sigma}_{2L} \\ \vdots & \vdots & \ddots & \vdots \\ \sqrt{\rho_L\rho_1}\, \bm{\Sigma}_{L1} & \sqrt{\rho_L\rho_2}\, \bm{\Sigma}_{L2} & \dots & \rho_L\bm{\Sigma}_{LL} \end{pmatrix},$ } \label{convergence}
\end{equation}

\noindent where $\rho_\ell = \lim_{n\rightarrow\infty} n_\ell/n$ and $\bm{\Sigma}_{\ell m} = \lim_{n\rightarrow\infty} \cov(\sqrt{n_\ell} \hat{G}_\ell(\theta_0), \sqrt{n_m} \hat{G}_m(\theta_0))$. This is an elementary but key intermediate result for establishing that the vector of GMM estimates $(\sqrt{n}(\hat\theta_\ell-\theta_0))_{\ell=1}^L$ is asymptotically normal. 

Conventional cluster-robust methods require independent clusters. Small-$L$ cluster-robust methods exploit the weaker requirement of {\em asymptotic independence}:
\begin{equation*}
  \bm{\Sigma}_{\ell m}=\zero \quad\forall\, \ell\neq m.\footnote{Assumption 2.2(i) of \cite{canay2020wild} directly imposes this requirement. Assumption 3.1(ii) of \cite{canay2017randomization} imposes symmetry of the limit distribution, which, under the group of transformations considered in their \S 4 and our equation \eqref{sinfp}, corresponds to this requirement.}
\end{equation*}

\noindent {\em We therefore interpret this as the key requirement for the validity of cluster-robust methods.} Our goal is to provide primitive restrictions on the clusters and data-generating process under which this holds. Toward this end, we next introduce some standard definitions from spectral graph theory \citep{chung1997spectral,trevisan2016lecture}.
\begin{definition} \hfill
  \begin{enumerate}[(a)]
    \item The {\em edge boundary size} of $S \subseteq \mathcal{N}_n$ with respect to $\bm{A}$ is 
      \begin{equation*}
	\abs{\partial_{\bm{A}}(S)} = \sum_{i \in S} \sum_{j\in \mathcal{N}_n\backslash S} A_{ij},
      \end{equation*}

      \noindent the number of links involving a unit in $S$ and a unit not in $S$.

    \item The {\em volume} of $S$ is $\text{vol}_{\bm{A}}(S) = \sum_{i\in S} \sum_{j=1}^n A_{ij}$, the sum of the {\em degrees} $\sum_{j=1}^n A_{ij}$ of units $i$ in $S$. 

    \item The {\em conductance} of $S$ (assuming it has at least one link) is
      \begin{equation*}
	\phi_{\bm{A}}(S) = \frac{\abs{\partial_{\bm{A}}(S)}}{\text{vol}_{\bm{A}}(S)}.
      \end{equation*}
  \end{enumerate}
\end{definition}

Conductance is a $[0,1]$ measure of how integrated $S$ is within $\bm{A}$. The denominator is a trivial upper bound on the numerator since it is possible to have all units in $S$ only connected to units outside of $S$. In this case, $S$ is maximally integrated with the rest of the network, and its conductance is one. On the other hand, any component has zero conductance and is therefore maximally segregated.

Our main condition for guaranteeing $\bm{\Sigma}_{\ell m}=\zero$ for all $\ell \neq m$ is 
\begin{equation}
  \max_{1\leq \ell \leq L} \phi_{\bm{A}}(\C_\ell) \rightarrow 0 \quad \text{as} \quad n\rightarrow\infty. \label{conductance}
\end{equation}

\noindent This says {\em the maximal conductance of the clusters is small}, which means each cluster's boundary size is of smaller order than its volume. Under additional conditions, \autoref{sufficiency} establishes sufficiency and \autoref{necessity} necessity. 

\bigskip 
\noindent {\bf Intuition.} At first glance, \eqref{conductance} may seem obvious; if links are primarily within, rather than between, clusters, then clusters appear to be mostly independent. However, this alone is insufficient because if two completely connected clusters are connected by only a single link, then all units are at most three links apart. Hence, we also require a restriction on the network's density for the clusters to be close to independent.

Real-world networks are often sparse \citep{barabasi2015}, meaning that most units are linked with a vanishingly small fraction of potential alters, or more formally, that the average degree is asymptotically bounded. Sparsity is implicitly required by our weak dependence conditions (\autoref{apsi}). We interpret sparsity as imposing an analog of increasing domain asymptotics (typical in spatial econometrics) in the sense that units must be minimally spaced apart if they are unlinked to most alters. For example, it rules out a completely connected network in which all units are linked, which is the densest possible topology. 

Now, if \eqref{conductance} holds, this means that, for two clusters of order $n$ size, the number of links connecting the clusters is $o(n)$. Consequently, given that units in the network are minimally spaced apart, most units in one cluster will be far from units in the other for large $n$. Such units are close to independent since correlation decays with distance under weak network dependence. Only units near the cluster boundaries (boundary units are those linked with units outside the cluster) are strongly dependent, but these have a negligible contribution under \eqref{conductance}. Hence, the clusters are approximately independent.\footnote{This is the same intuition for spatially \citep{bester2011inference} or temporally \citep{ibragimov2010t} correlated data, but it was previously unclear how this extended to network data. In particular, the network analogs of increasing domain asymptotics and the definition of ``boundary'' were not obvious.} As for necessity, as far as we are aware, there are no analogous prior results, but the intuition is straightforward. If \eqref{conductance} fails, then we could have, for example, each unit in one cluster directly linked to a unit in the other, in which case the clusters can be strongly dependent.

Note that \eqref{conductance} does not imply that the giant component asymptotically fractures into $L$ disconnected subnetworks. It requires the number of cross-cluster links to be small relative to within-cluster links, but the former is allowed to diverge with $n$. As a result, low-conductance clusters are not necessarily visually apparent. \autoref{RGGRCM} plots two random graphs with the same average degree from simulations in \autoref{ssimspec}, coloring units by clusters obtained via spectral clustering (see \autoref{sconstruct}). The maximal conductance in both panels is nearly identical (0.114 for the left panel and 0.116 for the right) even though clusters are more visually apparent in the left panel.

\begin{figure}
  \centering
  \begin{subfigure}{.5\textwidth}
    \centering
    \includegraphics[scale=0.35]{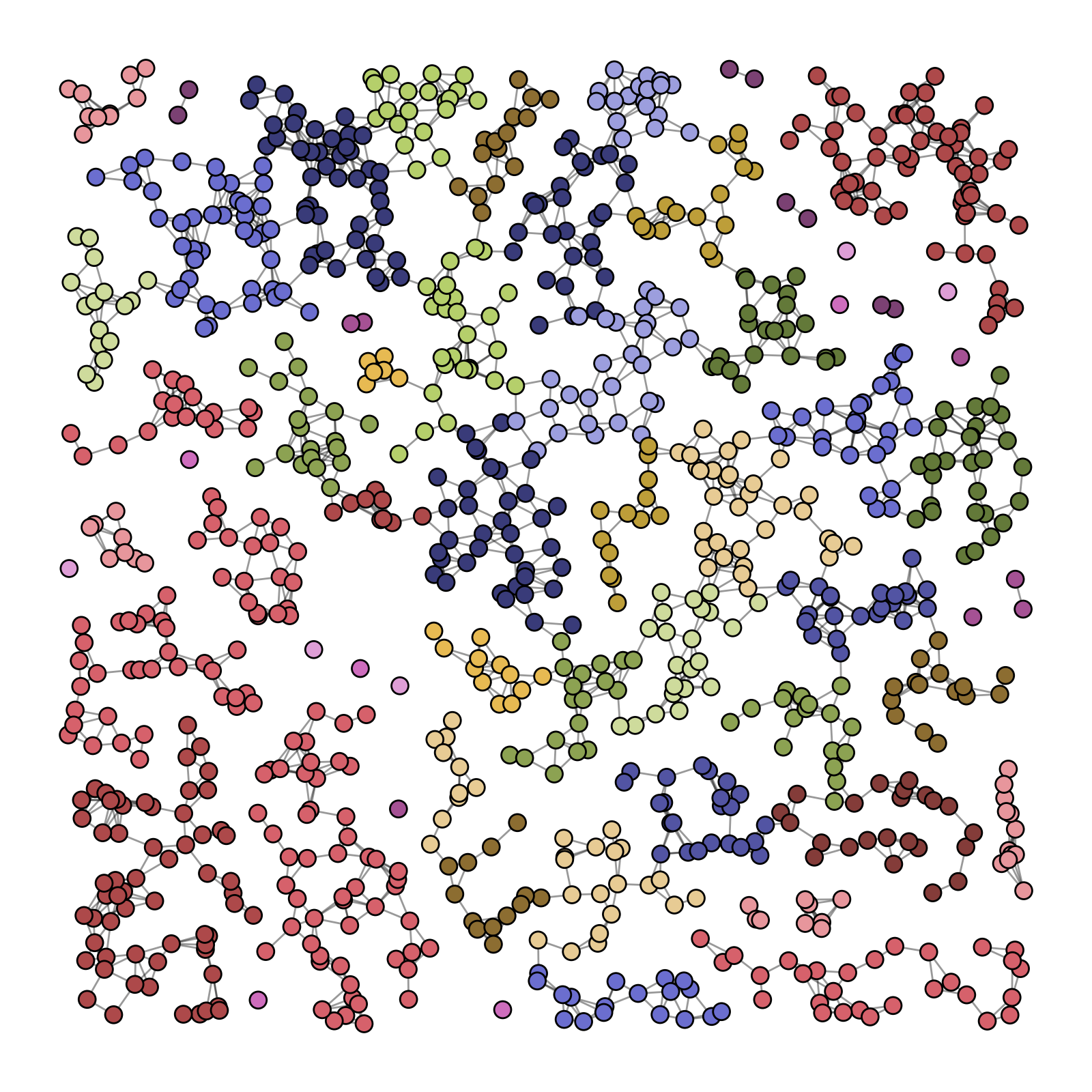}
  \end{subfigure}%
  \begin{subfigure}{.5\textwidth}
    \centering
    \includegraphics[scale=0.35]{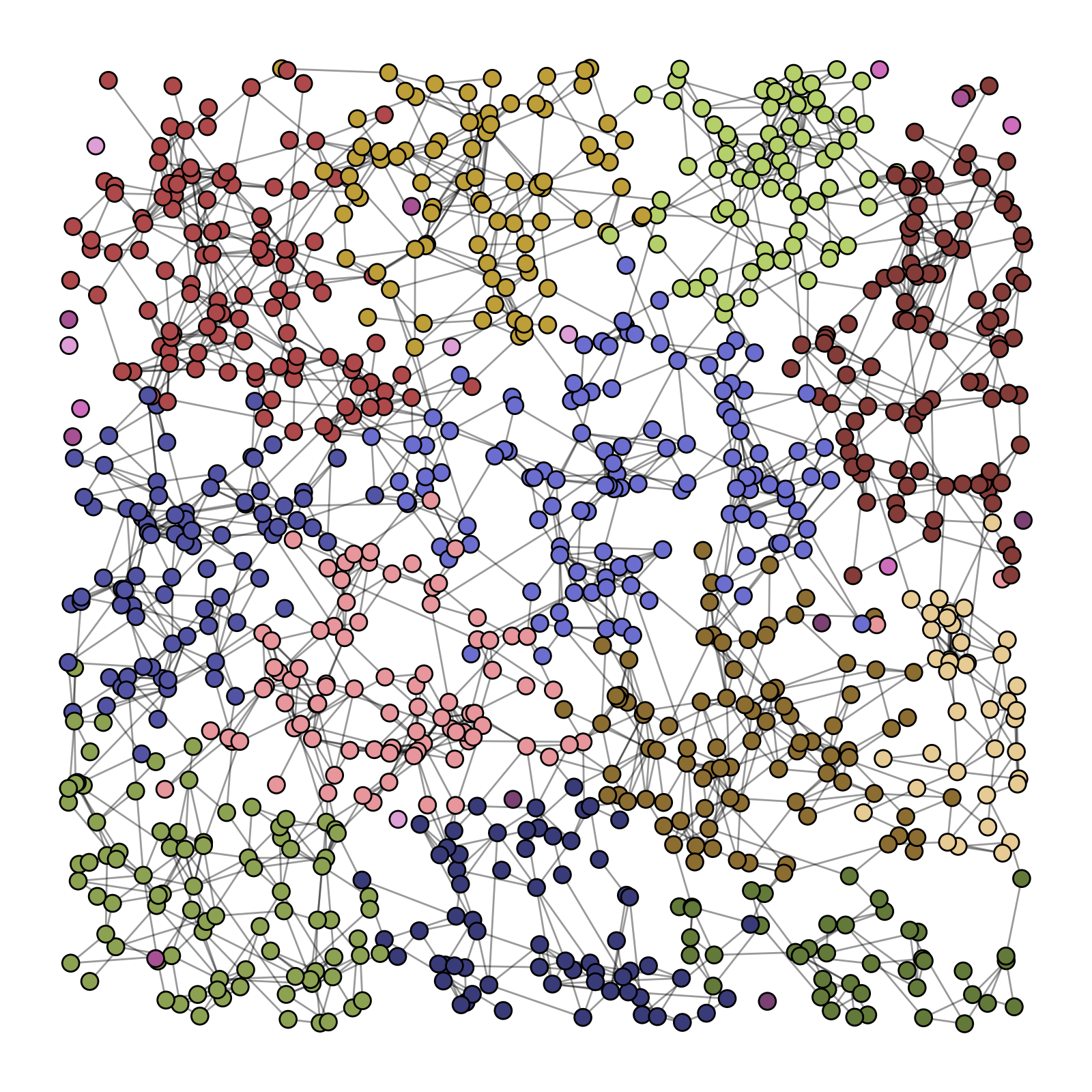}
  \end{subfigure}
  \caption{Low-conductance clusters of a random geometric graph (left panel) and a draw from the random connections model (right panel), obtained by spectral clustering. See \autoref{sspectra} for a description of these models. Some colors are reused for different clusters.}
  \label{RGGRCM}
\end{figure}

\subsection{Graph Invariants}\label{sginvs}

The next natural question is whether a sequence of clusters satisfying \eqref{conductance} necessarily exists for any given network sequence. In the spatial setting, one can always construct clusters satisfying the required boundary rate condition under increasing domain asymptotics \citep[][Assumption 2(iv)]{bester2011inference}, for example by partitioning $\R^2$ into rectangles, but this turns out not to be true in general for networks.

For any integer $k>1$, define the {\em $k$th-order Cheeger constant} of $\bm{A}$ 
\begin{equation}
  h_k(\bm{A}) = \min\left\{\max_{1\leq \ell \leq k} \phi_{\bm{A}}(S_\ell)\colon S_1, \dots, S_k  \text{ partitions } \mathcal{N}_n\right\}. \label{cheegercon}
\end{equation}

\noindent This is a well-known graph invariant that measures network integration. It is simply the lowest maximal conductance over all possible partitions of size $k$. For example, if $\bm{A}$ has $k$ components, then it is maximally segregated, and $h_k(\bm{A})=0$. 

Clearly, a necessary condition for \eqref{conductance} is
\begin{equation}
  h_L(\bm{A}) \rightarrow 0.
  \label{chee}
\end{equation}

\noindent Thus, if $h_L(\bm{A})$ could be computed, it would provide a simple way to assess whether low-conductance clusters exist. Furthermore, the argmin would be the best possible partition for cluster-robust inference. Unfortunately, this optimization problem is NP-complete \citep{vsima2006np}.

Fortunately, the spectrum of the graph Laplacian, which can be efficiently computed, is highly informative about the magnitude of the Cheeger constant. Let $\bm{D}$ be the $n\times n$ diagonal matrix with $ii$th entry equal to $i$'s degree $\sum_{j=1}^n A_{ij}$. The {\em (normalized) Laplacian} of $\bm{A}$ is
\begin{equation*}
  \bm{I} - \bm{D}^{-1/2}\bm{A}\bm{D}^{-1/2},
\end{equation*}

\noindent where $\bm{I}$ is the $n\times n$ identity matrix. Order the eigenvalues of the Laplacian as
\begin{equation*}
  \lambda_1 \leq \lambda_2 \leq \ldots \leq \lambda_n.
\end{equation*}

\noindent The following facts are well known: $\lambda_k \in [0,2]$ for all $k$, and $\lambda_k=0$ if and only if $\bm{A}$ has at least $k$ components (hence $\lambda_1=0$) \citep{chung1997spectral}. 

The second fact suggests that, if $\lambda_k$ is close to zero, then $\bm{A}$ should contain a set of $k$ clusters with low conductance since having $k$ components is the ``ideal'' case of $k$ zero-conductance clusters. The (higher-order) Cheeger inequality formalizes this intuition by relating Cheeger constants to the spectrum of the Laplacian as follows:
\begin{equation}
  \frac{\lambda_k}{2} \leq h_k(\bm{A}) \leq C \lambda_k^{1/2}, \label{cheegin}
\end{equation}

\noindent where $C$ is a constant that does not depend on $n$ and is $O(k^2)$ \citep[][Theorem 1.1]{lee2014multiway}.\footnote{\cite{chung1997spectral} and \cite{trevisan2016lecture} provide simpler proofs for $k=2$. The lower bound holds because, by the variational characterization of $\lambda_k$, one can rewrite $\lambda_k$ as the optimum of an objective that constitutes a continuous relaxation of the optimization problem corresponding to $h_k(\bm{A})$.} This yields the following simple corollary.

\begin{proposition}\label{easy}
  $h_L(\bm{A}) \rightarrow 0$ if and only if $\lambda_L \rightarrow 0$.
\end{proposition}
\begin{proof}
  The ``if'' direction is immediate from both inequalities in \eqref{cheegin}. The ``only if'' direction follows from the first inequality and the fact that $\lambda_k \in [0,2]$ for all $k$.
\end{proof}

The proposition gives us a feasible way of assessing \eqref{chee}, which is to examine the magnitude of $\lambda_L$. Based on \eqref{chee} and the proposition, we make the following definition.

\begin{definition}
  A sequence of networks is {\em well clustered at $L$} if $\lambda_L \rightarrow 0$ as $n\rightarrow\infty$. It is {\em well clustered} if it is well clustered at $L$ for some $L>1$ and {\em poorly clustered} otherwise.
\end{definition}

\noindent That is, well-clustered networks can be partitioned into multiple low-conductance clusters. As we discuss in \autoref{sspectra}, there are reasonable models of network formation that produce well-clustered networks, but there are also models that produce poorly clustered networks. It is therefore important to assess whether a given network is well clustered in practice.

\subsection{Practical Recommendations}\label{srecs}

Based on our results, we make several recommendations for empirical practice. We provide a Python package for implementing them, linked in \autoref{networkinference}.

\bigskip
\noindent {\bf Conductance.} Given a candidate set of clusters $\{\C_\ell\}_{\ell=1}^L$, however it is obtained, one should compute its maximal conductance $\max_{1\leq \ell \leq L} \phi_{\bm{A}}(\C_\ell) \in [0,1]$. Our asymptotic theory shows this quantity must be small relative to $n$. Based on simulations in \autoref{smc}, we recommend using clusters with maximal conductance at most 0.1 to ensure that cluster-robust methods have adequate size control.

\bigskip
\noindent {\bf Laplacian.} To construct clusters in the first place, the first task is to select the number of clusters $L$.\footnote{By Theorem 3 of \cite{ibragimov2010t}, there is no data-dependent way to optimally select $L$ while maintaining size control uniformly in the number of clusters.}
Based on our simulations, we recommend choosing the largest $L$ such that $\lambda_L$ is at most $0.05$ to ensure that clusters have sufficiently small conductance.\footnote{A widely used heuristic for spectral clustering is to select $L$ such that $\lambda_L$ is ``small'' and $\lambda_{L+1}$ is ``large,'' if such an $L$ exists \citep{von2007tutorial}. This is because, by \autoref{easy}, $L$ low-conductance clusters exist but $L+1$ such clusters do not in this case. A similar heuristic is used in principal components analysis.} The reason for choosing the {\em largest} such $L$ is that power requires multiple clusters. The results in \cite{cai2022implementation} and \cite{cameron2008bootstrap} suggest at least five clusters are required for nontrivial power because, for the case of the randomization test, there only exist $2^L$ values in the randomization distribution used to construct the critical value in \eqref{sinfp}. 

\bigskip
\noindent {\bf Computing clusters.} Given $L$, the ideal set of clusters is the partition that minimizes conductance. While an exact solution is computationally infeasible, many feasible alternatives have been developed. We focus on spectral clustering in our simulation results. In \autoref{sconstruct}, we state the algorithm and discuss why it computes low-conductance clusters when they exist.

Note that we are not interested in interpreting the clusters themselves, as is usually the case when applying community detection algorithms, but rather in finding an approximate solution to the infeasible objective of minimizing conductance. One can therefore use any number and combination of algorithms and heuristics to best achieve this goal, so long as the maximal conductance is small.\footnote{Provided these algorithms only use the network data $\bm{A}$, as is the case for most community detection algorithms, our asymptotic theory remains valid since it treats $\bm{A}$ as fixed.}

For some networks, community detection algorithms return an {\em unbalanced partition} consisting of one large and several very small clusters; see \autoref{einhom} in \autoref{sspectra}. We view such networks as being close to poorly clustered. Result \eqref{convergence} implies that only large clusters (of order $n$ size) contribute to the limit distribution, so this situation is little better than having only one cluster, in which case cluster-robust methods have trivial power. {\em Consequently, unless at least five large clusters with low conductance can be found, we recommend against using cluster-robust methods.} 

An alternative is the HAC variance estimator, which is described in \autoref{smc}. The upside of HAC is that theoretically, it is consistent for the variance regardless of the number of low-conductance clusters, while the downside is that confidence intervals constructed using HAC variance estimators are often quite anti-conservative in smaller samples, as reflected in the simulations in \autoref{smc}.

\bigskip
\noindent {\bf Inference procedure.} If $L$ is fairly large, conventional cluster-robust methods such as \cite{liang1986longitudinal} clustered standard errors may be used. Otherwise, the small-$L$ cluster-robust methods cited in \autoref{sinferp} exhibit better performance. As noted by \cite{cameron2015practitioner}, there is no clear-cut definition of ``fairly large,'' but $L$ being ``small'' may range from being less than 20 to less than 50.

\begin{remark}[Weighted, directed graphs]\label{rwdg}
  Our asymptotic results rely on a CLT that only pertains to binary, undirected networks, but as discussed in \cite{kojevnikov2021bootstrap}, this can be extended to weighted networks. The definitions of conductance, the Cheeger constant, and the Laplacian immediately apply to weighted ($A_{ij}\in \R$) graphs and have been generalized to directed ($A_{ij}\neq A_{ji}$) graphs. Likewise, the Cheeger inequality applies directly to weighted graphs \citep{lee2014multiway} and has been extended to directed graphs \citep{chung2005laplacians}. Consequently, we believe that our recommendations are relevant to these types of networks.
\end{remark}

\subsection{Constructing Clusters}\label{sconstruct}

To construct clusters, ideally, we would like to find the solution to \eqref{cheegercon}, but as discussed in \autoref{sginvs}, this is not computationally feasible. This partly motivates a large, multi-disciplinary literature on network clustering algorithms. 

\cite{zacchia2020knowledge} uses a popular modularity-based algorithm due to \cite{blondel2008fast} to construct clusters. Such algorithms seek to find a partition that approximately maximizes an alternative measure of community structure known as ``modularity,'' which is given by
\begin{equation*}
  Q = \frac{1}{2m} \sum_{\ell=1}^L \sum_{i,j \in \C_\ell} \left( A_{ij} - \frac{d_id_j}{2m} \right),
\end{equation*}

\noindent where $d_i$ is $i$'s degree and $m$ is the number of links in $\bm{A}$. The sum over $A_{ij}$ alone is the number of within-community links, while the sum over $d_id_j / (2m)$ is the expected number of within-community links under a null configuration model. If $Q$ is large, it means there exists a ``surprising'' number of within-community links under this partition, which is indicative of community structure, so we seek to choose the partition that maximizes $Q$. Modularity-based algorithms are the subject of a large literature in computer science and physics \citep[e.g.][Ch.\ 9]{barabasi2015}. 

Spectral clustering algorithms are a popular method with theoretical guarantees \citep{von2007tutorial}. Given a desired number of clusters $L$, these algorithms apply $k$-means or some other clustering method to the $L$ eigenvectors of the Laplacian, a typical implementation being the following.
\begin{enumerate}
  \item Given a graph $\bm{A}$ and desired number of clusters $L$, compute the Laplacian and its eigenvalues $\lambda_1 \leq \ldots \leq \lambda_n$.

  \item Let $V_\ell$ be the eigenvector associated with $\lambda_\ell$ and $V_{\ell i}$ its $i$th component. Embed the $n$ units in $\R^L$ by associating each unit $i$ with a position
    \begin{equation*}
      X_i = \left(\frac{V_{1i}}{(\sum_{\ell=1}^L V_{\ell i}^2)^{1/2}}, \dots, \frac{V_{Li}}{(\sum_{\ell=1}^L V_{\ell i}^2)^{1/2}}\right).
    \end{equation*}

  \item Cluster the positions $\{X_i\}_{i=1}^n$ using $k$-means with $k=L$ to obtain $\C_1, \dots, \C_L$.
\end{enumerate}

\noindent As discussed in \cite{von2007tutorial}, this can be roughly interpreted as a continuous relaxation of the ideal program \eqref{cheegercon}.

There are well-known theoretical results showing that spectral clustering recovers low-conductance clusters, provided they exist. These are supported by the simulation evidence in \autoref{ssimspec}. Recall that if $\lambda_L=0$, the network consists of $L$ components, which are ``ideal'' clusters with exactly zero conductance. This intuitively suggests that if $\lambda_L$ is close to zero, then the network has $L$ low-conductance clusters. 

To see how this intuition translates to the eigenvectors, consider an ``ideal'' network consisting of $L$ components. The eigenvector $V_L$ associated with $\lambda_L$ then almost perfectly identifies the clusters because it can be written as $V_L = \bm{D}^{1/2} V_L^*$ for some $V_L^*$ with the property that $V_{Li}^*=V_{Lj}^*$ if and only if $i,j$ are in the same component \citep{peng2017partitioning}. That is, up to degree scaling due to $\bm{D}^{1/2}$, units in the same component are assigned the same value by $V_L$, whereas units in different components are assigned different values. Recovering the clusters is then a simple task for $k$-means (the normalization in the definition of $X_i$ adjusts for the degree scaling).


Now suppose more realistically that the network has $L$th-order Cheeger constant that is small relative to $\lambda_{L+1}$ (their ratio tends to zero). This implies $\bm{A}$ has $L$ low-conductance clusters by the Cheeger inequality and a spectral gap in that $\lambda_L/\lambda_{L+1} \rightarrow 0$.\footnote{In \autoref{ssimspec}, we apply spectral clustering to networks with small spectral gaps that are nonetheless well clustered and find that the algorithm still delivers low-conductance clusters. Hence, having a large spectral gap is sufficient for finding low-conductance clusters but apparently not necessary.} By Theorem 1.1 of \cite{peng2017partitioning}, the span of the $L$ eigenvectors of the Laplacian (corresponding to the smallest $L$ eigenvalues) is close to that of $L$ vectors of normalized indicators that identify the infeasible optimal partition that minimizes conductance. Consequently, the output of $k$-means should be close to the optimum. 

\section{Empirical Illustrations}\label{sapp}

To illustrate our results, we next revisit the analyses of two papers that cluster standard errors on subnetworks output by community detection algorithms.

\bigskip
\noindent {\bf \cite{aral2017exercise}.} This paper finds evidence of peer effects in exercise activity using data from an online social network of 1.1 million runners. The authors partition their network into 15144 clusters (average size 7.7, SD 41) using a modularity-based method. The online social network they study spans many countries, with US users comprising 20 percent of the data. They report that the giant component of their network contains 90 percent of units, so the large number of clusters they obtain is not because their network has many small components. A natural question is whether 15144 is a reasonable number of clusters.

Our results provide guidance for assessing cluster quality using the conductance measure. While the authors' network data is not publicly available, on p.\ 35 of the supplementary appendix, they write, ``on average 8 out of 10 friends are within-cluster while 2 of 10 are across clusters.'' They report this statistic based on their belief that it measures independence of clusters, and our results provide some formal justification for this belief. Their statistic is similar to conductance, being an average over a measure of conductance defined at the unit level, so we can ballpark $\max_{1\leq\ell\leq L} \phi_{\bm{A}}(\C_\ell)$ at around 0.2 in their application. 

Our results in \autoref{smc} indicate that cluster-robust methods would perform better if this were closer to 0.1. Fortunately, since the authors have 15k clusters, it is likely that they could settle for a fraction of 15k to obtain a substantial decrease in conductance while still leaving a sizeable number for power.

\bigskip
\noindent {\bf \cite{zacchia2020knowledge}.} This paper studies knowledge spillovers across firms. The author constructs a weighted, undirected network of 707 firms for each year $t$. The weighted link $A_{ij,t}$ between firms $i$ and $j$ at time $t$ measures co-patenting between firm inventors. In order to apply a community detection algorithm, which requires a static network, the author sums the networks across time, defining $A_{ij} = \sum_t A_{ij,t}$ for each $i,j$. 

To compute the clusters, the author applies a variant of the Louvain algorithm to the giant component, which is a modularity-based method \citep{blondel2008fast}. A tuning parameter $\varphi$ controls the number of clusters, and the choice of $\varphi=0.6$ used in the paper yields 20 clusters in the giant. The author treats all units outside of the giant as a single cluster. 

Our theoretical results in \autoref{smain} only pertain to unweighted graphs. However, the graph invariants in \autoref{sginvs} are all defined for weighted graphs, as discussed in \autoref{rwdg}. We compute these quantities both for the original weighted network and the unweighted version where $A_{ij}$ is set to 1 if and only if the weight is positive. In the unweighted graph, there are 3451 links, so the network is sparse. The analysis that follows focuses on the giant component, which consists of 439 units.

\begin{figure}
  \centering
  \includegraphics[scale=0.6]{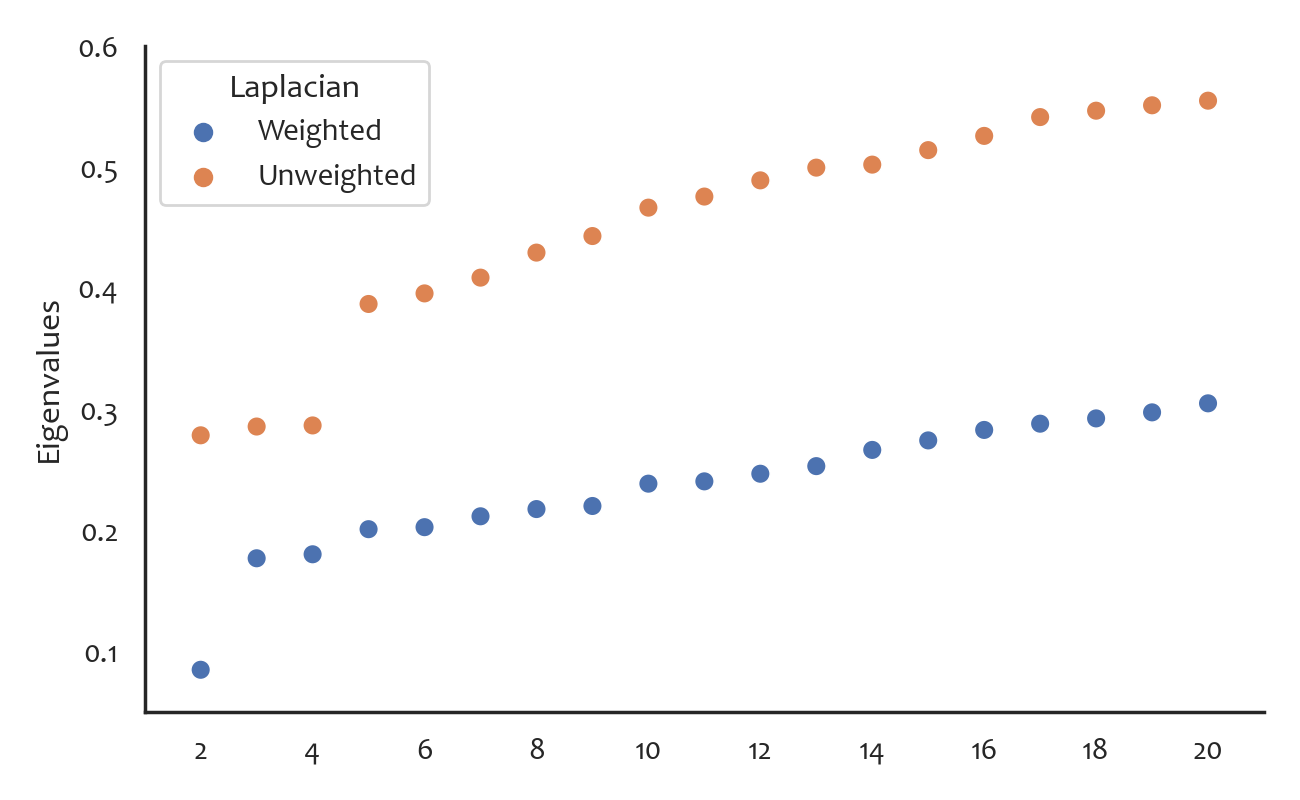}
  \caption{Scatterplot of eigenvalues $\lambda_2 \leq \ldots \leq \lambda_{20}$.}
  \label{app_spectra}
\end{figure}

\begin{figure}
  \centering
  \includegraphics[scale=0.5]{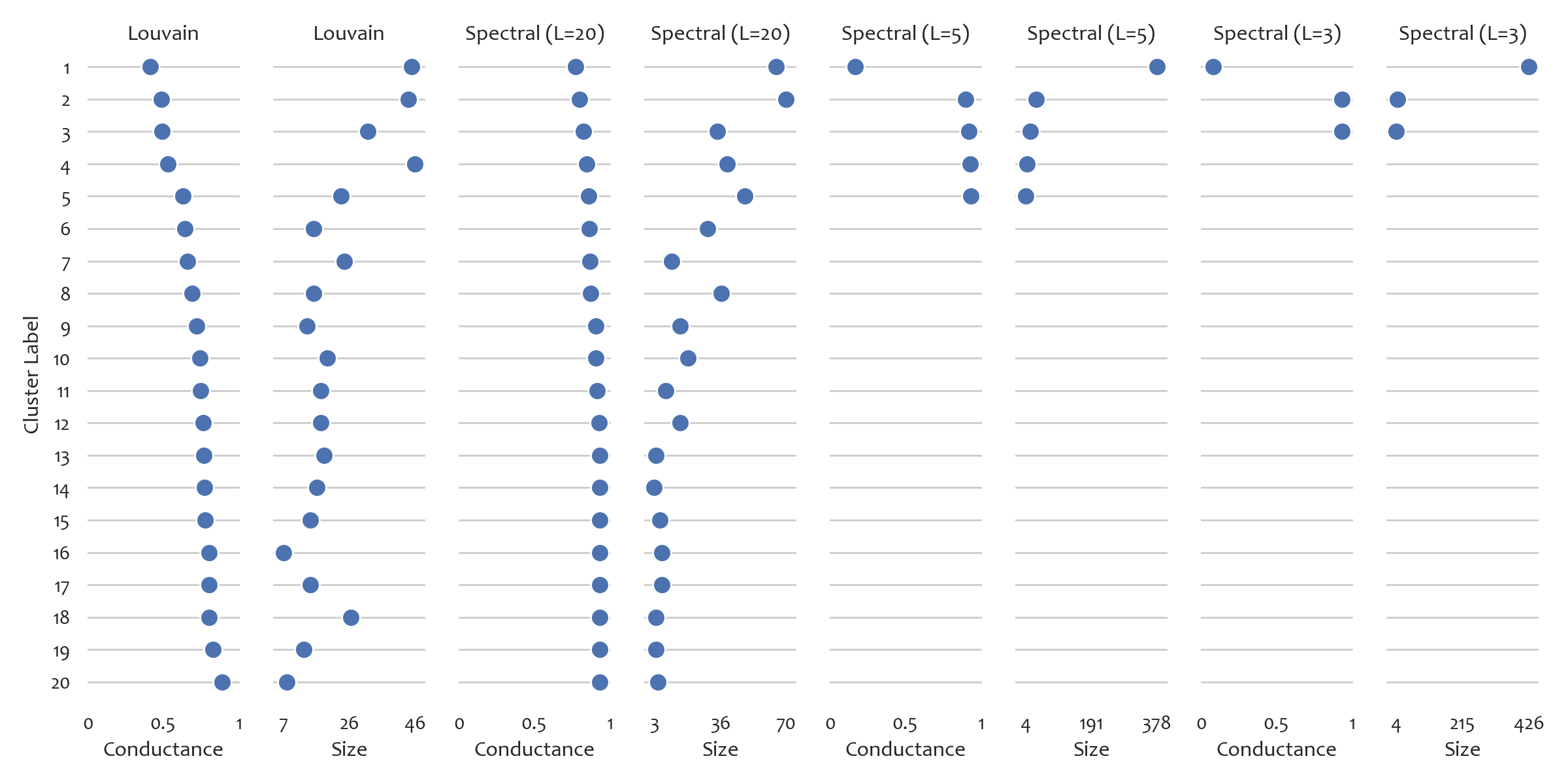}
  \caption{Conductances and cluster sizes for the weighted network.}
  \label{weighted}
\end{figure}

\autoref{app_spectra} plots the spectra of the Laplacians of both the weighted and unweighted networks, starting at $\lambda_2$ (since $\lambda_1=0$). The networks exhibit spectral gaps at 2 and 4. However, only one eigenvalue is below 0.1 in the weighted network and none are in the unweighted network, so they appear to be poorly clustered.

This is further confirmed in \autoref{weighted}. The first two columns plot the conductances and sizes of the 20 clusters used in \cite{zacchia2020knowledge} for the weighted network, and the remaining columns plot the same quantities for clusters obtained via spectral clustering, for different values of $L$. The corresponding figure for the unweighted network is essentially the same and therefore omitted. The figure shows that the Louvain algorithm for $\varphi=0.6$ and the spectral clustering algorithm for $L=20$ both yield high-conductance clusters. Choosing smaller values of $L$ for spectral clustering does not improve matters, as only a single cluster with somewhat low conductance emerges, but this contains the vast majority of units in the network. Thus, the partition is highly unbalanced, which, as discussed in \autoref{srecs}, means the power of the test would be little better than if we had only one cluster. 

Our results in \autoref{smc} show that cluster-robust methods can suffer from severe size distortion when networks are poorly clustered. Thus, for this dataset, it may be preferable to use an alternative inference procedure.

\section{Asymptotic Theory}\label{smain}

We consider a sequence of networks and associated clusters, both implicitly indexed by the network size $n$. Recall that $n_\ell = \abs{\C_\ell}$, the size of cluster $\ell$.

\begin{assump}[Limit Sequence]\label{aseq}
  (a) The number of clusters $L$ is fixed as $n\rightarrow\infty$. (b) For any $\ell=1,\dots,L$, $n_\ell/n \rightarrow \rho_\ell \in [0,1]$. 
\end{assump}

\noindent We consider a small number of clusters in part (a) because this appears to be typical of networks, as is the case in \autoref{sspectra}. Part (b) defines $\rho_\ell$ as the asymptotic fraction of units in $\C_\ell$, allowing for the possibility that the cluster has trivial size ($\rho_\ell=0$).

\cite{kojevnikov2021limit} provide a formal notion of weak network dependence called $\psi$-{\em dependence}, which we employ in what follows. This is analogous to familiar notions of spatial or temporal weak dependence, except the distance between observations is measured using path distance. Their paper and \cite{leung2020causal} verify $\psi$-dependence for a variety of network models used in practice.

Weak dependence simply means the correlation between two sets of observations decays as the network distance between the sets grows. Formalizing this requires some notational overhead. For any $H,H' \subseteq \mathcal{N}_n$, define $G_H = (g(W_i,\theta_0))_{i \in H}$ and $\ell_{\bm{A}}(H,H') = \min\{\ell_{\bm{A}}(i,j)\colon i \in H, j \in H'\}$, the distance between the two subsets. Let $\mathcal{L}_d$ be the set of bounded $\R$-valued Lipschitz functions on $\R^d$, $\norm{f}_\infty = \sup_x \abs{f(x)}$, $\text{Lip}(f)$ the Lipschitz constant of $f \in \mathcal{L}_d$, and
\begin{equation*}
  \mathcal{P}_n(h,h';s) = \left\{ (H,H')\colon H,H' \subseteq \mathcal{N}_n, |H|=h, |H'|=h', \ell_{\bm{A}}(H,H') \geq s \right\},
\end{equation*}

\noindent the set of subset pairs $H,H'$ with respective sizes $h,h'$ that are at least distance $s$ apart in the network. Define $i$'s $s$-neighborhood boundary $\N_{\bm{A}}^\partial(i,s) = \{j\in\N_n\colon \ell_{\bm{A}}(i,j) = s\}$ and its $k$th moment $\delta_n^\partial(s; k) = n^{-1} \sum_{i=1}^n \abs{\N_{\bm{A}}^\partial(i,s)}^k$. Finally, let
\begin{align*}
  &\Delta_n(s, m; k) = \frac{1}{n} \sum_{i=1}^n \max_{j\in\N_{\bm{A}}^\partial(i,s)} \abs{\N_{\bm{A}}(i,m) \backslash \N_{\bm{A}}(j,s-1)}^k \quad\text{and} \\ 
  &c_n(s, m; k) = \inf_{\alpha>1} \Delta_n(s, m; k\alpha)^{1/\alpha} \delta_n^\partial(s; \alpha/(\alpha-1))^{1-1/\alpha}. 
\end{align*}

\begin{assump}[Weak Network Dependence I]\label{apsi} \hfill
  \begin{enumerate}[(a)]
    \item There exist a constant $C>0$ and uniformly bounded constants $\{\psi_n(s)\}_{s, n \in \mathbb{N}}$ with $\psi_n(0)=1$ for all $n$ such that $\sup_n \psi_n(s) \rightarrow 0$ as $s\rightarrow\infty$ and
      \begin{equation*}
	\abs{\cov(f(G_H), f'(G_{H'}))} \leq C h h' (\norm{f}_\infty + \text{Lip}(f))(\norm{f'}_\infty + \text{Lip}(f')) \psi_n(s)
      \end{equation*}

      \noindent for all $n,h,h' \in \mathbb{N}$; $s>0$; $f \in \mathcal{L}_{d_g h}$; $f' \in \mathcal{L}_{d_g h'}$; and $(H,H') \in \mathcal{P}_n(h,h';s)$.

    \item There exist $p>4$ and a sequence $m_n\rightarrow\infty$ such that, for any $k \in \{1,2\}$, 
      \begin{equation*}
	\frac{1}{n^{k/2}} \sum_{s=0}^\infty c_n(s, m_n; k) \psi_n(s)^{1-(2+k)/p} \rightarrow 0 \quad\text{and}\quad n^{3/2} \psi_n(m_n)^{1-1/p} \rightarrow 0. 
      \end{equation*}
  \end{enumerate}
\end{assump}

\noindent This imposes weak network dependence on the set of moments.\footnote{It can be verified given an analogous weak network dependence condition imposed on the data $\{W_i\}_{i=1}^n$ and smoothness conditions on $g(\cdot)$ \citep[][Appendix A.1]{kojevnikov2021limit}.} Part (a) encodes the definition of $\psi$-dependence and Assumption 2.1 of \cite{kojevnikov2021limit}. Part (b) is Condition ND of the same reference and can be verified using arguments in \S A.1 of \cite{leung2020causal}. The key quantity is $\psi_n(s)$, which essentially bounds the correlation between sets of observations at distance $s$ and is required to decay to zero as $s$ diverges. \cite{leung2020causal} provides conditions under which $\psi_n(s)$ is uniformly $O(e^{-c\,s})$ for some $c>0$ in well-known social interactions models.

Part (b) is analogous to mixing conditions for spatial data that require the mixing coefficient to decay sufficiently quickly. For network data, this is necessarily more complicated to state because the metric space is non-Euclidean. As discussed in \cite{leung2020causal}, whereas the number of units in a ball of radius $s$ grows at most polynomially with $s$ in Euclidean space, it can grow exponentially in a network. Part (b) requires $\psi_n(s)$ to decay sufficiently quickly relative to the growth rates of $s$-neighborhood boundaries (analogs of balls in Euclidean space), as measured by $c_n(s, m_n; k)$. This is conceptually the same requirement underlying spatial mixing conditions.

\begin{assump}[Regularity]\label{areg}
  (a) $\cov(\sqrt{n_\ell} \hat{G}_\ell(\theta_0), \sqrt{n_m} \hat{G}_m(\theta_0)) \rightarrow \bm{\Sigma}_{\ell m}$, a finite matrix, for any $\ell,m = 1,\dots,L$, and $\bm{\Sigma}_{\ell\ell}$ is positive definite for any $\ell$. (b) For $p$ in \autoref{apsi}(b), $\sup_{n\in\mathbb{N}} \max_{i\in\mathcal{N}_n} \E[\norm{g(W_i,\theta_0)}^p] < \infty$.
\end{assump}

\begin{proposition}\label{basic}
  Under Assumptions \ref{aseq}--\ref{areg}, \eqref{convergence} holds. 
\end{proposition}

\noindent The proof of the proposition and all other results in this section are stated in \autoref{sproofs}. Under additional standard regularity conditions, we can use \autoref{basic} to establish asymptotic normality of $(\sqrt{n}(\hat\theta_\ell-\theta_0))_{\ell=1}^L$, but since this type of result is well-known, we omit these conditions and the corresponding result. 

We next state conditions required for our first main result, which establishes that the off-diagonal blocks of $\bm{\Sigma}^*$ in \eqref{convergence} are zero. Combined with \autoref{basic}, this verifies the key high-level condition required by small-$L$ cluster-robust inference methods. Let $\delta(\bm{A}) = n^{-1} \sum_{i=1}^n \sum_{j=1}^n A_{ij}$, the average degree. 

\begin{assump}[Conductance]\label{acond}
  $\max_{1\leq \ell \leq L} \phi_{\bm{A}}(\C_\ell) \cdot \delta(\bm{A}) \rightarrow 0$.
\end{assump}

\noindent This requires the conductance of all clusters to vanish as $n$ grows since $\delta(\bm{A})$ is generally bounded away from zero (otherwise, the network would be empty in the limit). In dense networks, $\delta(\bm{A})$ diverges, in which case the assumption requires the maximal conductance to shrink to zero faster. However, with nontrival network dependence, \autoref{apsi} generally requires sparsity in the sense $\delta(\bm{A}) = O(1)$.

Define $\mathcal{N}_{\bm{A}}(i,s) = \{j\in\mathcal{N}_n\colon \ell_{\bm{A}}(i,j) \leq s\}$, the {\em $s$-neighborhood} of unit $i$, and $M_n(s,k) = n^{-1} \sum_{i=1}^n \abs{\mathcal{N}_{\bm{A}}(i,s)}^k$, the $k$th moment of the $s$-neighborhood size.

\begin{assump}[Weak Network Dependence II]\label{apsi2}
  For $p$ in \autoref{apsi}(b) and some $\epsilon>0$, $\sum_{s=1}^n s\, M_n(s,2(1+\epsilon))^{1/(1+\epsilon)} \psi_n(s)^{1-2/p} = O(1)$.
\end{assump}

\noindent This is conceptually the same as \autoref{apsi}(b), requiring dependence to decay quickly enough relative to the growth rate of $s$-neighborhood sizes. We verify the condition in \autoref{sveri}. Note that the validity of the HAC estimator \citep[][Proposition 4.1]{kojevnikov2021limit} does not require Assumptions \ref{acond} or \ref{apsi2}; instead, it requires conditions relating the bandwidth and kernel to the network topology.

\begin{theorem}[Sufficiency]\label{sufficiency}
  Under Assumptions \ref{aseq}--\ref{apsi2}, $\sqrt{\rho_\ell\rho_m}\,\bm{\Sigma}_{\ell m}=\zero$ for all $\ell\neq m$ (whether or not $\rho_\ell$ and $\rho_m$ are nonzero).
\end{theorem}

\noindent This justifies the use of cluster-robust methods for network data. Lemma 1 of \cite{bester2011inference} is the analogous result for spatial data. Their proof makes extensive use of properties of Euclidean space, whereas our setting is non-Euclidean. 

\begin{theorem}[Necessity]\label{necessity}
  There exists a data-generating process for $\{g(W_i,\theta_0)\}_{i=1}^n$ and sequence of sparse networks ($\delta(\bm{A}) = O(1)$) satisfying Assumptions \ref{apsi}, \ref{areg}, and \ref{apsi2} such that, for any sequence of clusters $\{\C_\ell\}_{\ell=1}^L$ that satisfies $\min_\ell \rho_\ell > 0$ (nontrivial cluster sizes), $\liminf_{n\rightarrow\infty} \min_\ell n_\ell^{-1} \sum_{i\in\C_\ell} \sum_{j=1}^n A_{ij} >0$ (nonempty subnetworks), and \autoref{aseq}, yet fails to satisfy \autoref{acond}, we have $\sqrt{\rho_\ell\rho_m}\,\bm{\Sigma}_{\ell m} \neq \zero$ for some $\ell,m = 1, \dots, L$ with $\ell\neq m$.
\end{theorem}

\noindent The proof constructs a simple $\psi$-dependent process for which $\bm{\Sigma}^*$ is not block-diagonal. Since the process satisfies Assumptions \ref{aseq}--\ref{areg} and \ref{apsi2} but not \autoref{acond}, this establishes the latter's necessity under weak conditions on the clusters.

\section{Spectra of Geometric and Random Graphs}\label{sspectra}

By \autoref{necessity}, \eqref{conductance} is necessary for cluster-robust inference to be valid. As discussed in \autoref{ssum}, a necessary condition for \eqref{conductance} is that the network is well clustered, meaning for some $L$, $\lambda_L \rightarrow 0$, or equivalently, $h_L(\bm{A}) \rightarrow 0$ by \autoref{easy}. This section shows that, unfortunately, not all networks are well clustered, which motivates the recommendations in \autoref{srecs}. We first discuss results from geometry and random graph theory on the spectra and Cheeger constants of various graphs. We then provide simulation evidence supporting the theory and clarifying aspects that are, to our knowledge, incomplete.

\subsection{Theoretical Results}

The first two examples are well-clustered graphs. 

\begin{example}
  A planar graph is a graph that can be drawn on the plane such that links do not cross. \cite{kelner2011metric} show that planar graphs with uniformly bounded degree satisfy $\lambda_L = O(L/n)$. More generally, they show that for graphs embedded in orientable surfaces of genus $g$, the same result holds if $g$ does not depend on $L$ or $n$. 
\end{example}

\begin{example}\label{ergg}
  Random geometric graphs are defined by associating each unit $i$ with a position $X_i \in \R^d$, i.i.d.\ across units with density $f$, and setting $A_{ij} = \mathbf{1}\{\norm{X_i-X_j} \leq r_n\}$ for some $r_n > 0$. For the graph to be sparse, $r_n$ must tend to zero. Several papers characterize the limiting behavior of Cheeger constants. \cite{muller2020optimal} show that the second-order Cheeger constant $h_2(\bm{A})$ is $o(1)$ a.s.\ when $r_n \rightarrow 0$ and $nr_n^d / \log n \rightarrow\infty$.\footnote{This is their Theorem 2.1 for $v=2$, $b=1$. Their notion of conductance corresponds to our definition multiplied by the link count $\sum_{i,j} A_{ij}$, what they label $\text{Vol}_{n,2}(\mathcal{X}_n)$. By their equation (2.12), that term is of exact order $n^2 r_n^d$. Furthermore, the limit in (2.11) is finite. Hence, the normalization in their Theorem 2.1 implies the Cheeger constant is $O(r_n) = o(1)$ a.s.} 
  Theorem 12 of \cite{trillos2016consistency} provides a similar result for $L$th-order Cheeger constants, albeit defined slightly different than ours.
\end{example}

Both examples have $\lambda_L\rightarrow 0$, but the models are rather stylized. A perhaps more realistic model is the random connections model
\begin{equation}
  A_{ij} = \mathbf{1}\{\alpha_i + \alpha_j + r_n^{-1}\norm{X_i-X_j} > \varepsilon_{ij}\}, \label{rcm}
\end{equation}

\noindent which allows units further than distance $r_n$ to form links but with probability vanishing with distance $\norm{X_i-X_j}$. To our knowledge, there are no available results on the spectrum, but we provide simulation evidence below showing that this graph is well clustered.

The remaining examples show that not all graphs are well clustered.

\begin{example}
  A $k$-regular graph has the property that $\sum_j A_{ij}=k$ for all $i$. \cite{bollobas1988isoperimetric} proves that for $k\geq 3$, the isoperimetric number of a randomly drawn $k$-regular graph is at least a certain positive constant with probability approaching one. The isoperimetric number for $k$-regular graphs equals $k h_2(\bm{A})$, $k$ times the second-order Cheeger constant, so the latter is bounded away from zero.

  Expander graphs are those that, by construction, have Cheeger constants uniformly bounded away from zero, yet may still be sparse \citep{hoory2006expander}. 
\end{example}

Thus, for expander and $k$-regular graphs, $\lambda_2$ is bounded away from zero. However, these are extremely stylized models. The final example concerns more realistic models that have been applied to real-world networks.

\begin{example}\label{einhom}
  Inhomogeneous random graphs \citep{bollobas_phase_2007} satisfy
  \begin{equation*}
    \prob(A_{ij}=1 \mid \alpha_i, \alpha_j) = \frac{\kappa(\alpha_i,\alpha_j)}{n},
  \end{equation*}

  \noindent where typically the types $\alpha_i$ are assumed independent and $\kappa(\alpha_i,\alpha_j)$ assumed to have bounded support. Stochastic block models correspond to the special case in which types are finitely supported. These are widely used in the statistics literature for studying community detection. 

  One can easily choose $\kappa(\cdot)$ to generate homophily in types, meaning units with similar types have a higher probability of linking. In this case, it may seem that these graphs can be well clustered under reasonable conditions, with clusters roughly corresponding to sets of units of the same type. However, when $\kappa(\alpha_i,\alpha_j)$ has support bounded away from zero, $\bm{A}$ is stochastically bounded below by a nontrivial Erd\H{o}s-R\'{e}nyi graph (the special case with constant $\kappa(\cdot)$). According to \cite{hoffman2021spectral}, several papers studying Erd\H{o}s-R\'{e}nyi graphs ``show that the giant component can be partitioned into a well-connected expanding core together with small (logarithmic size) graphs attached to the core,'' where the core is a subgraph of the giant that is order $n$ in size.\footnote{See e.g.\ \cite{coja2007laplacian} (Theorem 1.2) for formal results for Erd\H{o}s-R\'{e}nyi graphs and \cite{zhang2018understanding} for related results on stochastic block models.} This suggests that even if $\lambda_L$ were small for these graphs, any low-conductance partition would be extremely unbalanced, so at best the graph would be close to poorly clustered. Our simulations below support this.
\end{example}

\subsection{Simulation Results}\label{ssimspec}

We simulate the random geometric graph (RGG) (\autoref{ergg}), random connections model (RCM) \eqref{rcm}, Erd\H{o}s-R\'{e}nyi graph (ER), and stochastic block model (SBM) (\autoref{einhom}) for $n=1000$ units. We calibrate parameters to obtain an average degree of about 5 for all graphs. For the RGG, $\{X_i\} \stackrel{iid}\sim \mathcal{U}([0,1]^2)$, and $r_n = (5/(\pi n))^{1/2}$. For the RCM, $\{X_i\}$ is drawn from the same distribution, $\{\alpha_i\} \stackrel{iid}\sim \mathcal{U}([0,1])$, $\{\varepsilon_{ij}\} \stackrel{iid}\sim$ logistic, $\{X_i\} \indep \{\alpha_i\} \indep \{\varepsilon_{ij}\}$, and $r_n = (5/(3.5\pi n))^{1/2}$. For ER, $\{A_{ij}\} \stackrel{iid}\sim \text{Ber}(5/n)$. Finally, for the SBM, we construct 10 blocks of 100 units each, where units in the same block have probability $40/n$ of linking and units in different blocks have probability $(5 \cdot 4/9)/n$ of linking. This results in an expected degree of six, with four links within cluster on average.

We analyze the spectrum of and apply spectral clustering to the subnetwork on the giant component. \autoref{sim_spectra} plots histograms and scatterplots of the spectra for a typical draw from each model. We see that both the RGG and RCM have a sizeable mass of eigenvalues near zero. In contrast, the ER and SBM each have only one zero eigenvalue ($\lambda_1$ is necessarily zero) and a large gap between $\lambda_1$ and $\lambda_2$. Consequently, only the RGG and RCM appear to be well clustered.

\begin{figure}
  \centering
  \includegraphics[scale=0.6]{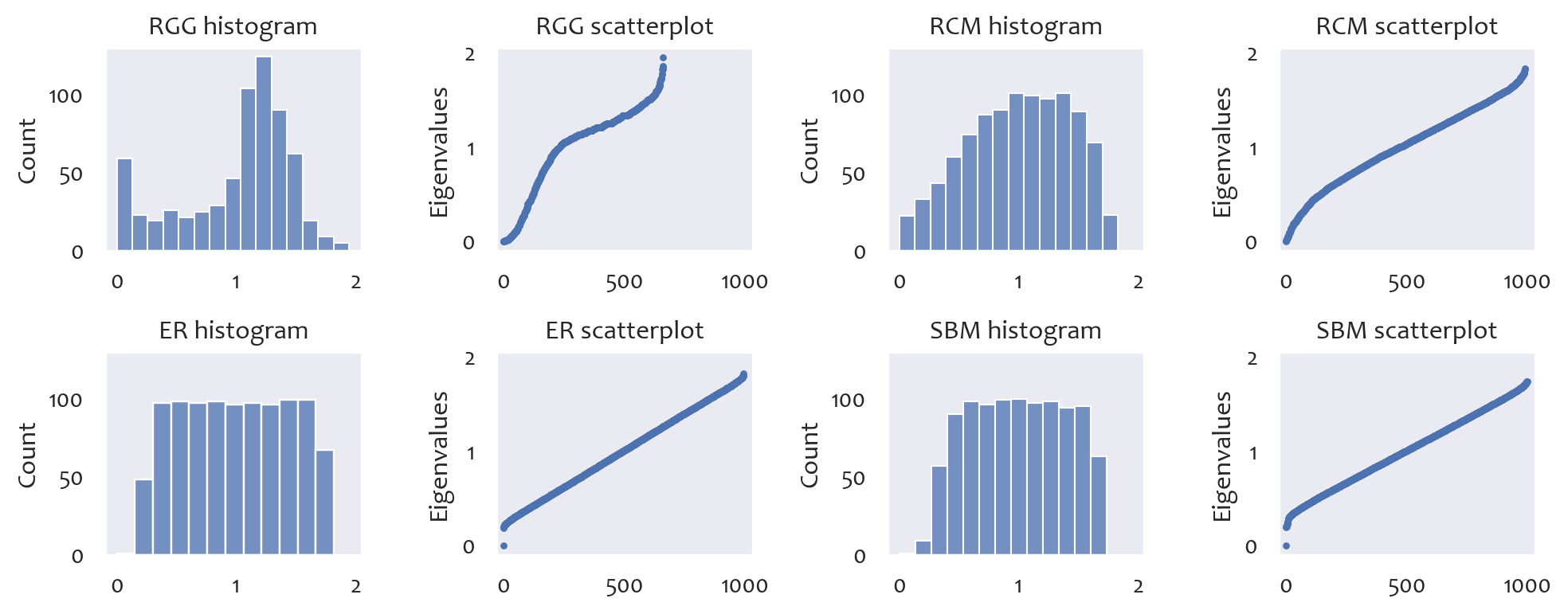}
  \caption{Histograms and scatterplots of eigenvalues.}
  \label{sim_spectra}
\end{figure}

We next present results on the conductance of clusters generated from spectral clustering, choosing the largest $L$ such that $\lambda_L < 0.05$. \autoref{snc} displays the result of 10k simulations. The first column of the table is the maximal conductance, the second the number of clusters in the giant component, the third the size of the spectral gap $\lambda_{L+1}-\lambda_L$, the fourth the $L$th smallest eigenvalue, the fifth the median cluster size, the sixth the size of the giant component, and the last the average degree. We find a larger number of clusters for the RGG and a moderate number for the RCM, both with small conductance. For ER and the SBM, the giant component is essentially never subdivided into smaller clusters. 

\autoref{RGGRCM} plots a single simulation draw for the RGG and RCM. Units are plotted according to their positions in $[0,1]^2$ and colored according to their clusters obtained from spectral clustering. The RCM contains longer-range links, producing a denser-looking figure, so spectral clustering generates fewer clusters compared to the RGG.

\begin{table}[ht]
\centering
\caption{Spectra and clusters}
\begin{threeparttable}
\begin{tabular}{lrrrrrrr}
\toprule
{} &  $\max_\ell \phi(\C_\ell)$ &  \# Clus. &  Gap &  $\lambda_L$ &  Med.\ Clus. &  Giant &  Degree \\
\midrule
RGG &                      0.158 &    40.0 & 0.004 &        0.048 &       18.3 & 758.4 &   4.83 \\
RCM &                      0.137 &    12.7 & 0.006 &        0.047 &       77.3 & 983.9 &   4.98 \\
ER  &                      0.000 &     1.0 & 0.180 &        0.000 &      992.9 & 993.1 &   5.00 \\
SBM &                      0.000 &     1.0 & 0.187 &        0.000 &      997.6 & 997.6 &   5.96 \\
\bottomrule
\end{tabular}
\begin{tablenotes}[para,flushleft]
  \footnotesize $n=1$k. Averages over 10k simulations. ``Gap'' = size of spectral gap, ``Med.\ Clus.'' = median cluster size, ``Giant'' = size of giant component, ``Degree'' = average degree.
\end{tablenotes}
\end{threeparttable}
\label{snc}
\end{table}

\section{Simulation Study}\label{smc}

We present simulation results on the finite-sample properties of the randomization test \eqref{sinfp} with clusters computed using spectral clustering and the $t$-test using a HAC variance estimator. If $\hat\theta$ is a sample mean, as is the case below, the latter is given by
\begin{equation*}
  \frac{1}{n} \sum_{i=1}^n \sum_{j=1}^n (W_i - \hat\theta) (W_j - \hat\theta) K\big( b_n^{-1}\ell_{\bm{A}}(i,j) \big),
\end{equation*}

\noindent where $K\colon \R_+ \rightarrow \R_+$ is a kernel function and $b_n \in \R_+$ a bandwidth. We use the uniform kernel and choose $b_n$ according to equation (13) of \cite{leung2020causal}.

\bigskip
\noindent {\bf Design 1.} We begin with a deliberately simplistic design to show that, even here, cluster-robust methods can be quite oversized for poorly clustered networks. We simulate the RGG, RCM, and SBM using the same parameters as the design in \autoref{ssimspec}. The SBM always has ten blocks. We then draw $\{\varepsilon_i\}_{i=1}^n \stackrel{iid}\sim \mathcal{N}(0,1)$ independently of the network and define $W_i = \varepsilon_i + \sum_j A_{ij} \varepsilon_j / \sum_j A_{ij}$ to obtain a simple form of network dependence in which units are only dependent if $\ell_{\bm{A}}(i,j) \leq 2$. We let $\theta_0 = \E[W_i] = 0$ and $g(W_i,\theta_0) = W_i$, so the goal is inference on the mean of $W_i$.

We test the null that $\theta_0=0$ at the 5 percent level. For each simulation draw, we compute $L=8$ clusters in the giant component. We treat all other components as individual clusters and discard clusters with sizes less than 20, so the average $L$ can be slightly smaller or larger than 8. We will see that, with $L=8$, clusters of the RGG will have low conductance, so the randomization test should perform well. Clusters of the RCM will have higher conductance, and it is unclear whether this will translate to substantial size distortion. Finally, clusters of the SBM will have high conductance, so the test may perform poorly.

We report rejection rates for the randomization test and two different $t$-tests. One uses a HAC variance estimator and the other the naive i.i.d.\ variance estimator. The latter serves to quantify the degree of dependence in the data. 

\autoref{simreject} reports the results of 10k simulations. We see that the randomization test controls size well under the RGG, outperforming HAC in smaller samples. This is a result of the low maximal conductance of the clusters. More surprising is that the test has good performance under the RCM, despite the conductance being as high as 0.22, with the test again outperforming HAC in smaller samples. For the SBM, we see that the randomization test consistently exhibits size distortion across all sample sizes due to the high maximal conductance of the clusters. HAC is also oversized in smaller samples, but the size distortion tends to zero as $n$ grows. The contrast in size is even starker for more complex forms of dependence considered in the next design.

\autoref{powercurves} plots power curves for the randomization test and HAC under the RGG for 10k simulations, indicating that the latter is more powerful. Analogous plots for the RCM are similar and therefore unreported.

\begin{table}[ht]
\centering
\caption{Size under design 1}
\begin{threeparttable}
\begin{tabular}{lrrrrrrrrr}
\toprule
{} & \multicolumn{3}{c}{RGG} & \multicolumn{3}{c}{RCM} & \multicolumn{3}{c}{SBM} \\
$n$ &    250 &     500 & 1000 &   250 &     500 & 1000 &    250 &     500 & 1000 \\
\midrule
Rand                      &  0.052 &   0.049 &   0.051 &  0.058 &   0.058 &   0.053 &  0.088 &   0.081 &   0.084 \\
HAC                       &  0.070 &   0.061 &   0.058 &  0.076 &   0.065 &   0.056 &  0.092 &   0.069 &   0.056 \\
IID                       &  0.274 &   0.272 &   0.279 &  0.275 &   0.276 &   0.282 &  0.294 &   0.293 &   0.294 \\
$\max_\ell \phi(\C_\ell)$ &  0.104 &   0.046 &   0.028 &  0.219 &   0.141 &   0.094 &  0.381 &   0.372 &   0.367 \\
\# Clusters               &  8.886 &   9.483 &  10.275 &  7.994 &   8.000 &   8.000 &  7.401 &   7.690 &   7.904 \\
1st Cluster               & 62.2 & 110.3 & 197.1 & 55.0 & 109.9 & 223.1 & 52.0 & 112.0 & 237.9 \\
2nd Cluster               & 40.3 &  75.3 & 142.9 & 44.3 &  87.3 & 171.4 & 41.0 &  85.2 & 174.4 \\
$L$th Cluster             & 23.8 &  28.0 &  45.2 & 24.3 &  40.0 &  69.3 & 23.8 &  49.0 & 100.8 \\
\bottomrule
\end{tabular}
\begin{tablenotes}[para,flushleft]
  \footnotesize Averages over 10k simulations. The first three rows report the sizes of level-5\% tests. The last three rows report cluster sizes in descending order of size. 
\end{tablenotes}
\end{threeparttable}
\label{simreject}
\end{table}

\begin{figure}[ht]
  \centering
  \includegraphics[scale=1]{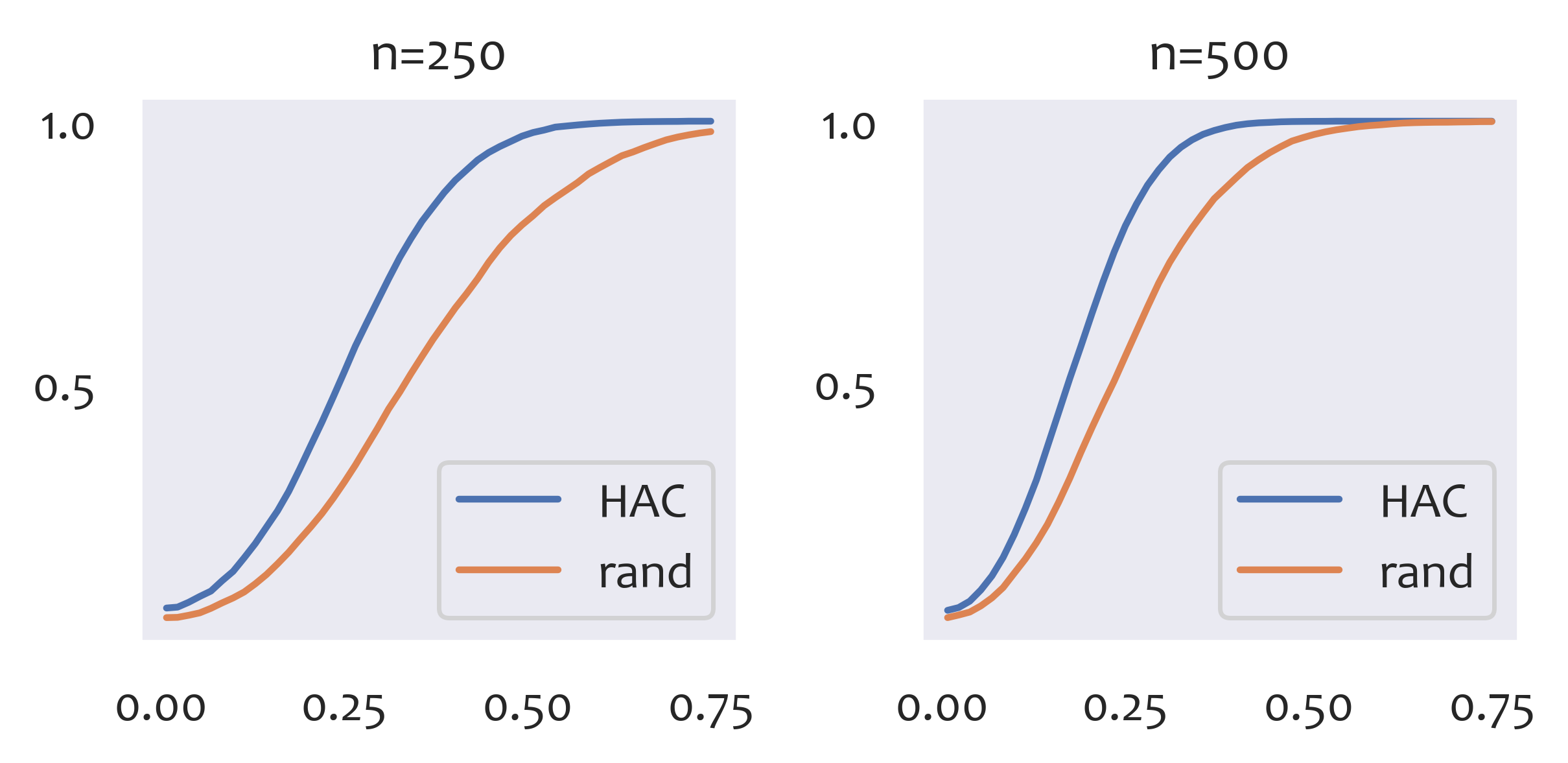}
  \caption{Power under design 1 for the RGG}
  \label{powercurves}
\end{figure}

\bigskip
\noindent {\bf Design 2.} The next design considers the more realistic problem of detecting network interference in a randomized experiment. We exactly replicate the designs in \S 5.2 of \cite{leung2020causal}, which involve two outcome models: a linear-in-means model and a binary game. For the former, $Y_i = V_i$, and for the latter, $Y_i = \ind\{V_i > 0\}$, where
\begin{equation*}
  V_i = \alpha + \beta \frac{\sum_j A_{ij}Y_j}{\sum_j A_{ij}} + \delta \frac{\sum_j A_{ij}D_j}{\sum_j A_{ij}}  +  D_i\gamma + \varepsilon_i.
\end{equation*}

\noindent Here $Y_i$ is unit $i$'s outcome, $\varepsilon_i$ its structural error, and $D_i$ its binary treatment assignment, which is i.i.d.\ and independent of all other primitives. For details on the choices of parameters and distributions of primitives, see \cite{leung2020causal}. For both models, we estimate a nonparametric spillover effect using  
\begin{equation*}
  \hat\theta = \frac{1}{n} \sum_{i=1}^n W_i \quad\text{for}\quad W_i = Y_i \left( \frac{T_i}{p_i} - \frac{1-T_i}{1-p_i} \right),
\end{equation*}

\noindent where $T_i=\bm{1}\{\sum_j A_{ij}D_j > 0\}$ and $p_i = \E[T_i]$ is the known propensity score. 

We simulate the two outcome models on three different networks. Two replicate those used in \cite{leung2020causal}, the configuration model and RGG, which are calibrated to the data on school friendship networks in his empirical application. The configuration model generates a network approximately uniformly at random from the set of all networks with a given degree sequence. That sequence is chosen to be the empirical degree sequence of the networks in his application, which has an average degree of about 8. This model plays the role of the SBM in \autoref{simreject}, being the network that lacks sufficient low-conductance clusters.

\begin{table}[ht]
\centering
\caption{Size under design 2}
\begin{threeparttable}
\begin{tabular}{lrrrrrrrrr}
\toprule
{} & \multicolumn{3}{c}{RGG} & \multicolumn{3}{c}{RCM} & \multicolumn{3}{c}{Configuration} \\
$n$       & 365 & 716 & 1408 & 365 & 722 & 1427 & 350 & 692 & 1375 \\
\midrule
LIM Rand         &   0.053 &   0.055 &    0.050 &   0.066 &   0.064 &    0.054 &         0.253 &   0.254 &    0.250 \\
LIM HAC          &   0.066 &   0.071 &    0.063 &   0.078 &   0.069 &    0.058 &         0.076 &   0.065 &    0.063 \\
TSI Rand         &   0.050 &   0.051 &    0.048 &   0.062 &   0.049 &    0.056 &         0.157 &   0.164 &    0.166 \\
TSI HAC          &   0.066 &   0.062 &    0.059 &   0.069 &   0.055 &    0.054 &         0.075 &   0.067 &    0.058 \\
$\max_S \phi(S)$ &   0.052 &   0.037 &    0.027 &   0.165 &   0.119 &    0.084 &         0.605 &   0.604 &    0.608 \\
\# Clusters      &   7.87 &   7.95 &    8.01 &   7.84 &   7.93 &    7.99 &         8.00 &   8.00 &    8.00 \\
1st Cluster      & 174.8 & 320.9 &  616.6 & 177.8 & 332.6 &  637.2 &       197.3 & 358.8 &  672.0 \\
2nd Cluster      & 141.2 & 252.6 &  467.1 & 134.7 & 243.4 &  456.3 &       119.0 & 210.2 &  390.3 \\
$L$th Cluster    &  61.0 & 113.9 &  213.3 &  69.0 & 128.5 &  243.9 &        79.8 & 147.4 &   21.0 \\
\bottomrule
\end{tabular}
\begin{tablenotes}[para,flushleft]
  \footnotesize Averages over 5k simulations. The first four rows report the sizes of level-5\% tests with LIM = linear-in-means, BG = binary game, Rand = randomization test, HAC = $t$-test with HAC estimator. The last three rows report cluster sizes in descending order of size. 
\end{tablenotes}
\end{threeparttable}
\label{rrdes2}
\end{table}

We additionally simulate the RCM using the same design as \autoref{simreject}. Like the RGG, the average degree is calibrated to the empirical application by setting $r_n = (\kappa/(3.5\pi n))^{1/2}$, where $\kappa$ is the average degree in the data. For both the RGG and RCM, we set the error terms $\varepsilon_i$ in the outcome models equal to $\nu_i + (X_{i1}-0.5)$, where $\nu_i$ is a normal error term and $X_{i1}$ is the first component of $i$'s position $X_i$ defined in \autoref{ssimspec}. The latter generates unobserved homophily. For additional details on the design, see \cite{leung2020causal}.

\autoref{rrdes2} reports rejection rates for 5k simulations for tests of the null that $\theta_0=\E[\hat\theta]$ equals its true value at the 5-percent level. We choose $L$ as in design 1 and discard all clusters with sizes less than 20. In the table, $n$ denotes the total population size across the largest, two largest, and four largest schools in the empirical application of \cite{leung2020causal}. The first two rows report results under the linear-in-means model and the next two under the binary game. 

Under the configuration model, which produces poorly clustered networks, the randomization test exhibits substantial size distortion, about five times the nominal level. However, the test is better sized than HAC for the other networks, which are well clustered. It exhibits some size distortion under the RCM when conductance exceeds 0.1 but remains better sized than HAC.

\section{Conclusion}\label{sconclude}

This paper studies the practice of partitioning a network, either manually or using an unsupervised learning algorithm, in order to apply cluster-robust inference methods. We isolate a key condition that, under mild assumptions, is necessary and sufficient for this practice to be valid: the clusters must all have low conductance (boundary-to-volume ratios). We call graphs ``well-clustered'' if a partition with this property exists and provide theoretical and simulation evidence showing that important classes of graphs are not well clustered. Our simulation study shows that cluster-robust inference methods, when applied to graphs that are well clustered, outperform HAC estimators in terms of size control. However, for graphs that are poorly clustered, cluster-robust methods can exhibit severe size distortion. Our results connect the literature on cluster-robust inference to spectral graph theory and spectral clustering, which provide simple methods for assessing whether a graph is well clustered and for constructing clusters in a more principled way, namely to minimize conductance. 

We provide three recommendations for empirical practice. First, for any candidate set of clusters, one should compute the maximal conductance, a measure of cluster quality, and aim for a value below 0.1. Second, given a network, one should first compute the spectrum of the Laplacian to ascertain the number of low-conductance clusters. Third, given a well-clustered network, one can construct a desired number of clusters using spectral clustering. If these methods produce at least five low-conductance clusters relatively balanced in size, then cluster-robust methods can perform well. Otherwise, an alternative procedure may be preferable.

\appendix
\numberwithin{equation}{section} 
\numberwithin{table}{section}

\section{Appendix}\label{sappendix}

\subsection{Proofs}\label{sproofs}

\begin{proof}[Proof of \autoref{basic}]
  Let $\mathcal{S}$ be the subset of clusters $\C_\ell$ for which $\rho_\ell > 0$. Then $n^{-1/2}(n_\ell \hat G_\ell(\theta_0))_{\ell \in \mathcal{S}}$ is asymptotically normal with the desired limit variance by the CLT of \cite{kojevnikov2021limit} (Theorem 3.2) and Cram\'{e}r-Wold device. For all $\ell$ such that $\rho_\ell = 0$, their CLT implies $n^{-1/2}n_\ell \hat G_\ell(\theta_0) \plimarrow \zero$, so we can extend joint convergence for only clusters in $\mathcal{S}$ to joint convergence for the full vector, as in \eqref{convergence}. 
\end{proof}

\bigskip
\begin{proof}[Proof of \autoref{sufficiency}]
  We show $n^{-1}\cov(n_\ell\hat{G}_\ell(\theta_0),n_m\hat{G}_m(\theta_0)) = o(1)$ for any $\ell\neq m$. Let $\norm{\cdot}$ be the matrix sup norm. The covariance is bounded in norm by
  \begin{equation}
    \frac{1}{n} \sum_{i \in \C_\ell} \sum_{j \in \C_m} \norm{\E[g(W_i,\theta_0) g(W_j,\theta_0)']} \leq C' \sum_{s=1}^n \psi_n(s)^{1-2/p} \frac{1}{n} \sum_{i \in \C_\ell} \sum_{j \in \C_m} \ind\{\ell_{\bm{A}}(i,j)=s\} \label{1st}
  \end{equation}

  \noindent for some constant $C'>0$ and $p$ in \autoref{apsi}(b). This follows from Corollary A.2 of \cite{kojevnikov2021limit} (choose $A=\{i\}$, $B=\{j\}$, $a=b=1$, and $q=p$), which we may apply due to Assumptions \ref{apsi} and \ref{areg}. The sum over $s$ terminates at $n$ because there are $n$ units.

  We next bound the term $n^{-1} \sum_{i \in \C_\ell} \sum_{j \in \C_m} \ind\{\ell_{\bm{A}}(i,j)=s\}$ in \eqref{1st}. Let 
  \begin{equation*}
    \mathcal{B}_{\bm{A}}(\C_\ell) = \big\{i\in\C_\ell\colon \max_{j\in\mathcal{N}_n\backslash\C_\ell} A_{ij}=1\big\},
  \end{equation*}
  
  \noindent the boundary of $\C_\ell$. Given $i \in \C_\ell$ and $j \in \C_m$ with $\ell_{\bm{A}}(i,j)=s$, there must exist some $k \in \mathcal{B}(\C_\ell)$ such that $\ell_{\bm{A}}(i,k)=d$ and $\ell_{\bm{A}}(k,j)=d'$, for some $d,d'$ satisfying $d+d'=s$ and $d'\geq 1$.
  Hence,
  \begin{multline}
    \frac{1}{n} \sum_{i \in \C_\ell} \sum_{j \in \C_m} \ind\{\ell_{\bm{A}}(i,j)=s\} \\\leq \frac{1}{n} \sum_{i=1}^n \sum_{j=1}^n \sum_{d=0}^{s-1} \sum_{k=1}^n \ind\{\ell_{\bm{A}}(i,k)=d\} \ind\{k \in \mathcal{B}_{\bm{A}}(\C_\ell)\} \ind\{\ell_{\bm{A}}(k,j)=s-d\} \\
    \leq \sum_{d=0}^{s-1} \frac{1}{n} \sum_{k=1}^n \ind\{k \in \mathcal{B}_{\bm{A}}(\C_\ell)\} \abs{\mathcal{N}_{\bm{A}}(k,d)} \,\abs{\mathcal{N}_{\bm{A}}(k,s-d)} \\
    \leq \sum_{d=0}^{s-1} \left( \frac{1}{n} \sum_{k=1}^n \abs{\mathcal{N}_{\bm{A}}(k,d)}^{1+\epsilon} \abs{\mathcal{N}_{\bm{A}}(k,s-d)}^{1+\epsilon} \right)^{1/(1+\epsilon)} \left( \frac{1}{n} \sum_{k=1}^n \ind\{k \in \mathcal{B}_{\bm{A}}(\C_\ell)\} \right)^{\epsilon/(1+\epsilon)} \label{2nd}
  \end{multline}

  \noindent for any $\epsilon>0$ by H\"{o}lder's inequality. Since 
  \begin{equation*}
    \frac{1}{n} \sum_{k=1}^n \abs{\mathcal{N}_{\bm{A}}(k,d)}^{1+\epsilon} \abs{\mathcal{N}_{\bm{A}}(k,s-d)}^{1+\epsilon} \leq \frac{1}{n} \sum_{k=1}^n \abs{\mathcal{N}_{\bm{A}}(k,s)}^{2(1+\epsilon)},
  \end{equation*}

  \noindent we have
  \begin{equation*}
    \eqref{2nd} \leq s \bigg( \underbrace{\frac{1}{n} \sum_{k=1}^n \abs{\mathcal{N}_{\bm{A}}(k,s)}^{2(1+\epsilon)}}_{M_n(s,2(1+\epsilon))} \bigg)^{1/(1+\epsilon)} \bigg( \underbrace{\frac{\abs{\mathcal{B}_{\bm{A}}(\C_\ell)}}{\text{vol}_{\bm{A}}(\C_\ell)}}_{\leq \phi_{\bm{A}}(\C_\ell)} \underbrace{\frac{\text{vol}_{\bm{A}}(\C_\ell)}{n}}_{\leq \delta(\bm{A})} \bigg)^{\epsilon/(1+\epsilon)} ,
  \end{equation*}

  \noindent noting that $\abs{\mathcal{B}_{\bm{A}}(\C_\ell)} \leq \abs{\partial_{\bm{A}}(\C_\ell)}$. Therefore,
  \begin{equation*}
    \eqref{1st} \leq C' \sum_{s=1}^n s\, M_n(s,2(1+\epsilon))^{1/(1+\epsilon)} \psi_n(s)^{1-2/p} \left( \max_{1\leq \ell \leq L} \phi_{\bm{A}}(\C_\ell) \cdot \delta(\bm{A}) \right)^{\epsilon/(1+\epsilon)}.
  \end{equation*}

  \noindent Choosing $\epsilon$ according to \autoref{apsi2}, this is $o(1)$ by Assumptions \ref{acond} and \ref{apsi2}.
\end{proof}

\bigskip
\begin{proof}[Proof of \autoref{necessity}]
  Consider a data-generating process such that, for some $\gamma,\gamma' > 0$ and $\gamma'' \geq 0$ and all $n$, $\E[g(W_i,\theta_0) g(W_j,\theta_0)] = \gamma \mathbf{1}\{i=j\} + \gamma' A_{ij} + \gamma'' (1-A_{ij})$ for all $i,j\in\mathcal{N}_n$ and $n$. Various restrictions on $\bm{A}$ and $\gamma'$ can ensure that Assumptions \ref{apsi} and \ref{apsi2} hold. A simple example is when $\gamma''=0$ and $\sum_j A_{ij}$ is uniformly bounded over $i$ and $n$, which yields a bounded-degree network sequence and locally dependent data where only linked observations are correlated. Various restrictions on $g(\cdot,\theta_0)$ can ensure \autoref{areg}, for example if it has bounded range.

  Since the network is sparse but \autoref{acond} fails, $\liminf_{n\rightarrow\infty} \max_{1\leq\ell\leq L} \phi_{\bm{A}}(\C_\ell) > 0$. Then, for some $\ell$, the following has positive limit infimum:
  \begin{equation*}
    \frac{\sum_{i\in\C_\ell} \sum_{j\not\in\C_\ell} A_{ij}}{\text{vol}_{\bm{A}}(\C_\ell)} = \sum_{m\neq \ell} \frac{\sum_{i\in\C_\ell} \sum_{j\in\C_m} A_{ij}}{\text{vol}_{\bm{A}}(\C_\ell)}.
  \end{equation*}

  \noindent This implies that for some cluster $m \neq \ell$,
  \begin{equation*}
    \liminf_{n\rightarrow\infty} \frac{\sum_{i\in\C_\ell} \sum_{j\in\C_m} A_{ij}}{\text{vol}_{\bm{A}}(\C_\ell)} > 0. 
  \end{equation*}
  
  \noindent For such $\ell,m$,
  \begin{multline*}
    \abs{n^{-1} \cov(n_\ell \hat{G}_\ell(\theta_0), n_m \hat{G}_m(\theta_0))} = \bigg| \frac{1}{n} \sum_{i \in \C_\ell} \sum_{j \in \C_m} \E[g(W_i,\theta_0) g(W_j,\theta_0)] \bigg| \\ \geq \gamma\, \frac{1}{n} \sum_{i \in \C_\ell} \sum_{j \in \C_m} A_{ij} = \gamma\, \frac{\sum_{i \in \C_\ell} \sum_{j \in \C_m} A_{ij}}{\text{vol}_{\bm{A}}(\C_\ell)} \frac{n_\ell}{n} \frac{1}{n_\ell} \sum_{i \in \C_\ell} \sum_{j=1}^n A_{ij}.
  \end{multline*}

  \noindent By assumption, $n_\ell / n \rightarrow \rho_\ell > 0$ and $\liminf_{n\rightarrow\infty} \min_\ell n_\ell^{-1} \sum_{i\in\C_\ell} \sum_{j=1}^n A_{ij} >0$. Therefore, the right-hand side of the above display has a strictly positive limit infimum.
\end{proof}

\subsection{Verifying Weak Network Dependence}\label{sveri}

\cite{leung2020causal} shows that $\psi_n(s)$ is uniformly $O(e^{-c\,s})$ for some $c>0$ for two well-known models of social interactions. His \S A verifies the working paper version of \autoref{apsi}(b) under the assumption that $\psi_n(s) = e^{-c\,s}$ for graphs with polynomial and exponential neighborhood growth rates, meaning
\begin{equation*}
  \max_{i\in\mathcal{N}_n}\, \abs{\mathcal{N}_{\bm{A}}(i,s)} = C s^d \quad\text{and}\quad \max_{i\in\mathcal{N}_n}\, \abs{\mathcal{N}_{\bm{A}}(i,s)} = C e^{\beta s},
\end{equation*}

\noindent respectively, for universal constants $C,d,\beta > 0$. In the polynomial case, no additional conditions are needed. In the exponential case, he requires $c > 3\beta$, meaning $\psi_n(s)$ decays sufficiently fast enough relative to the rate at which neighborhood sizes expand.

We verify \autoref{apsi2} under this setup. In the polynomial case, 
\begin{equation*}
  \sum_{s=1}^n s\, M_n(s,2(1+\epsilon))^{1/(1+\epsilon)} \psi_n(s)^{1-2/p} \leq C^2 \sum_{s=1}^\infty s^{2d+1} e^{-(1-2p^{-1}) c\,s} < \infty.
\end{equation*}

\noindent In the exponential case, 
\begin{equation*}
  \sum_{s=1}^n s\, M_n(s,2(1+\epsilon))^{1/(1+\epsilon)} \psi_n(s)^{1-2/p} \leq C^2 \sum_{s=1}^\infty s\,e^{(2\beta-(1-2p^{-1})c) s}.
\end{equation*}

\noindent This is finite if $(0.5-p^{-1})c > \beta$ for $p$ in \autoref{apsi}(b), which is slightly stronger than $c>3\beta$ since $p>4$ under \autoref{apsi}.


\FloatBarrier
\phantomsection
\addcontentsline{toc}{section}{References}
\bibliography{cluster}{} 
\bibliographystyle{aer}


\end{document}